\documentclass[%
 aip,
 amsmath,amssymb,
 reprint,%
]{revtex4-1}
\usepackage{amsmath}

\usepackage{graphicx}
\usepackage{dcolumn}
\usepackage{bm}
\usepackage{float}
\usepackage[utf8]{inputenc}
\usepackage[T1]{fontenc}
\usepackage{mathptmx}
\usepackage{hyperref}

\preprint{AIP/123-QED}

\makeatletter

\newcommand{\noun}[1]{\textsc{#1}}

\numberwithin{equation}{section}
\usepackage{amstext}
\usepackage{amsthm}
\theoremstyle{remark}
    \ifx\thechapter\undefined
      \newtheorem{rem}{\protect\remarkname}
    \else
      \newtheorem{rem}{\protect\remarkname}[chapter]
    \fi
\theoremstyle{plain}
    \ifx\thechapter\undefined
	    \newtheorem{thm}{\protect\theoremname}
	  \else
      \newtheorem{thm}{\protect\theoremname}[chapter]
    \fi
\theoremstyle{definition}
    \ifx\thechapter\undefined
      \newtheorem{example}{\protect\examplename}
    \else
      \newtheorem{example}{\protect\examplename}[chapter]
    \fi

\theoremstyle{definition}
    \ifx\thechapter\undefined
      \newtheorem{definition}{\protect\definitionname}
    \else
      \newtheorem{definition}{\protect\definitionname}[chapter]
    \fi

\makeatother

\DeclareMathOperator{\EX}{\mathbb{E}}

\providecommand{\examplename}{Example}
\providecommand{\remarkname}{Remark}
\providecommand{\theoremname}{Theorem}
\providecommand{\definitionname}{Definition}

\begin{document}
\title{Universal Upper Estimate for Prediction Errors under Moderate Model
Uncertainty}
\author{Bálint Kaszás}
\author{George Haller}\altaffiliation{Author to whom correspondence should be addressed:\\ georgehaller@ethz.ch}
\affiliation{Institute for Mechanical Systems, ETH Zürich,
Leonhardstrasse 21, 8092 Zürich, Switzerland}
\begin{abstract}
We derive universal upper estimates for model-prediction error under
moderate but otherwise unknown model uncertainty. Our estimates give
upper bounds on the leading order trajectory-uncertainty arising along model trajectories, solely as functions of the invariants of the known Cauchy-Green strain tensor of the model. Our bounds turn out to be optimal, which means that they cannot be improved for general systems. The quantity relating the leading-order trajectory-uncertainty to the model uncertainty is the Model Sensitivity, which we find to be a useful tool for a quick global assessment of the impact of modeling uncertainties in various domains of the phase space. Examining the expectation that Finite-Time Lyapunov Exponents capture sensitivity to modeling errors, we show that this does not generally follow. {However, we find that certain important features of the FTLE persist in the MS field.}
\end{abstract}

\maketitle
\begin{quotation}
We present a method of sensitivity analysis for general dynamical systems, subjected to deterministic or stochastic modeling uncertainty. Using the properties of the unperturbed dynamics, we derive a universal bound for the leading-order prediction error. This bound motivates the definition of the Model Sensitivity, a scalar quantity, depending on the initial condition and time. We demonstrate, using nonlinear numerical models that the Model Sensitivity provides both a global view over the phase space of the dynamical system and in some situations, a localized, time-dependent predictor of uncertainties along trajectories. We find that the phase-space structure of the Model Sensitivity (MS) is related, but not identical to that of the Finite-Time Lyapunov Exponents (FTLE). We formulate conditions under which robust features of the FTLE field are expected to also be seen in the MS field.
\end{quotation}
\section{\label{sec:Introduction}Introduction}

{One of the challenges in predicting and describing real-world phenomena
is uncertainties that enter the modeling process}.
Depending on the context, these can arise as a result of incomplete
or noisy data, uncertainty in the mathematical model, or even the
error introduced by numerical algorithms. Here we seek to bound the impact of these uncertainties on specific model trajectories utilizing minimal information on the modeling errors but substantial information on the internal dynamics of the known model. 

A range of approaches exist to assess the impact of model uncertainty. One such approach, response theory, originates
from statistical physics, where a central question was an equilibrium
system's response to infinitesimal perturbations. Under a (possibly
time-dependent) perturbation to a Hamiltonian system, Kubo's formula \citep{Kubo} establishes a link between the expected
value of the linear-order response and certain quantities of the unperturbed
system. This linear response theory was generalized to systems with uniform hyperbolicity. With this assumption, Ruelle's work put the theory on a rigorous
foundation, providing formulas for the asymptotic expansion of the invariant measure\cite{Ruelle2009a} of the perturbed system. {The results have been sharpened by the use of transfer operators\cite{butterley2007} and have been expanded to stochastic dynamical systems \cite{hairer2010}.}
Numerical
evidence shows \citep{Abramov2007,Majda2010,Lucarini2018} that linear
response can be observed even in systems that are not strictly uniformly
hyperbolic. In particular, in
the field of climate science, response theory has been successfully
used to assess the various possible scenarios of anthropogenic climate
change\citep{leith75,Gritsun2017,Lembo2020,bodai2020}. 

For general dynamical systems, an additional source of uncertainty
is also present: sensitivity to initial conditions. This means that a small error in the system's
initial condition grows exponentially, governed by the largest Lyapunov
exponent \citep{Nese1989a,Ginelli2013a}. A common illustration of this phenomenon is weather prediction, in which long
term predictions are impossible due to the exponential error-growth.
{Considerable effort has gone into assessing these difficulties,
for example, by using
ensemble methods\cite{leith75, Hawkins2009, maher2019, tel2019} to obtain a statistical characterization. Data assimilation\citep{Ott2004,Kalnay2012,Grudzien2018}, where model prediction is compared regularly to real observations, is also a prominent method of quantifying errors. It is also possible to quantify modeling errors, for example, by either introducing perturbations to both the background state and the observations\cite{cardinali2014} or by  adding a physically justified, stochastic forcing term to the model\cite{piccolo2016}.}

Another important question is the sensitivity
of model predictions to slight changes in the model parameters\cite{Laughton1964,Cacuci}. This sensitivity is often characterized
by the derivatives of an observable (a function of the model variables) with respect to those parameters\cite{Lea2000,Lea2002, Thuburn2005, Wang2013b, Wang2014}. {In general, sensitivity analysis seeks to assess prediction errors under modeling uncertainty. For this problem, response theory could also be employed\cite{ghil2020} successfully. Alternatively, one could also use finite differencing to approximate the derivative. However, to reduce the computational cost, the sensitivity is often computed
from the linearized dynamics (tangent method)\cite{Wang2013b,Wang2014} along a reference trajectory. }

Usually, the observed quantity is an infinite-time average computed along trajectories. Direct calculations need to utilize sufficiently long Monte Carlo simulations of the full model and finite difference approximations for the derivative.
In this case, the tangent method generally results in asymptotically unbounded sensitivities, which do not match the bounded ones computed directly\cite{Lea2000, Lea2002}. As noted in Ref. \onlinecite{Thuburn2005}, the issue comes from exchanging two limits: the sensitivity of an infinite-time average is the derivative of the infinite-time average, while the tangent method calculates the infinite-time average of a derivative. 

It has been suspected that similarly to sensitivity
with respect to initial conditions, sensitivity with respect to parameters is also
governed by the largest Lyapunov exponent of the underlying trajectory. This is supported
by numerical results\citep{Lea2000}, but can also be intuitively understood: the differential equation that describes the growth of perturbations
to initial conditions has the same homogeneous part as the one describing
error growth due to parameter changes \citep{Wang2013b,Wang2014}. {A connection between the two types of sensitivities has also been noted in the context of perturbation bounds of Markov chains\cite{mitrophanov2003, mitrophanov2006}. }
 
To circumvent this problem of unbounded averages, the ensemble method calculates the sensitivity over shorter time intervals for several randomly selected trajectories using the linearized dynamics. Then, the true sensitivity of the infinite-time average can be approximated by the ensemble average \citep{Lea2000, Lea2002}.

For ergodic systems, the least-squares shadowing method\citep{Wang2013b,Wang2014,Ni2016,Lasagna2019b} offers an alternative calculation of the true parameter sensitivity of an infinite-time average.
Instead of solving the tangent equation, the method looks for a nearby shadowing
trajectory which has a uniformly bounded distance from the reference
trajectory. Practically, this means that a nonlinear optimization
problem has to be solved. Solving the linearized version of this problem, it is possible to obtain meaningful
sensitivities\citep{Wang2014}, even for chaotic systems \citep{Wang2013b}.
The method was also implemented in turbulent fluid dynamical simulations\citep{Ni2016}.
Further improvements on calculating sensitivities for chaotic systems take advantage of unstable periodic orbits \citep{Lasagna2019b}.

In contrast to the methods mentioned above, we focus here on finite-time predictions and their uncertainties, as opposed to infinite-time averages. This is motivated by the fact that certain models may not be defined for infinite times, or the infinite time averages may not be accurate representations of the system\cite{tel2019}. 

{We derive universal bounds on the uncertainties in model predictions under small modeling errors. Our estimates only assume the knowledge of a general bound on the model errors, yet yield trajectory-specific bounds for the model-prediction errors.} These bounds provide a granular assessment of the impact of modeling errors, depending only on the known local dynamics of the phase space in the absence of model uncertainties. 
We relate the arising model sensitivities to the Finite Time Lyapunov Exponents (FTLE)\cite{Mathur2007} and their ridges\footnote{There exist various definitions of ridges in the literature\cite{eberly1994, lindeberg1998}, which are in general non-equivalent\cite{peikert2008}. Loosely speaking, we refer to locally maximizing hypersurfaces of the scalar field as ridges. To express this idea in more precise terms, we prefer to use the ridge definition in Ref. \onlinecite{Karrasch2013}, where a ridge is defined to be a structurally stable, attracting invariant manifold of the gradient-field with a structurally stable invariant boundary. This definition has the advantage that it is robust with respect to small perturbations of the scalar field. For the precise formulation, we refer to Definition 1.} , and hence to Lagrangian Coherent
Structures (LCS)\cite{Haller2015}, which are organizing structures in the idealized
model's phase space. {The model sensitivities are captured by a time-dependent scalar field, analogous to the FTLE.  We find that the ridges of the FTLE do not necessarily signal the presence of a ridge in the scalar field characterizing model sensitivity.} However, we formulate a plausible condition under which a correspondence is expected.

We also extend the analysis to cases
when both deterministic and stochastic uncertainties are present. We show that assuming
multiplicative Gaussian noise, the expectation value of the observation
error can be bounded by an asymptotic formula, analogous to the purely
deterministic case. All these estimates even turn out to be optimal: we give examples in which the inequalities become equalities. In addition, through numerical examples of models that represent
various levels of complexity, we show that the bounds developed for
the observation error hold for surprisingly large modeling uncertainties too. 

\section{\label{sec:Set-up}Set-up}

We first consider a parametrized family of deterministic differential equations
\begin{equation}
\dot{x}=f_{0}(x,t)+\varepsilon g(x,t,\varepsilon)\quad x\in U\subset\mathbb{R}^{n},\quad t\in[t_{0},t_{1}],\quad0\leq\epsilon\ll1\label{eq:startingeq2}
\end{equation}
where both $f_{0}$ and $g$ are assumed to be smooth functions of their arguments.
Trajectories of this equation are of the form $x(t;t_{0},x_{0},\varepsilon)$,
which are as smooth in their arguments as $f$ is. We can think of
$\varepsilon g(x,t;\varepsilon)$ as a family of perturbations representing errors
to a known model system
\begin{equation}
\dot{x}=f_{0}(x,t),\label{eq:model}
\end{equation}
our ``best understanding'' of the given problem. The perturbations of the form $\varepsilon g(x,t;\varepsilon)$ represent the modeling uncertainty of the underlying problem,
such as a systematic bias with spatial and temporal dependence. We assume that this term is bounded in norm. 

We are interested in how trajectories change under changes in the
parameter $\varepsilon.$ While the exact nature of the family
$\varepsilon g(x,t;\varepsilon)$ is generally unknown for $\varepsilon>0$, we still seek to assess the leading-order uncertainty of trajectories in case an overall bound on $\varepsilon |g(x,t;\varepsilon)|$ is available. We call this leading-order uncertainty the model-sensitivity of the trajectory with respect to the parameter
$\varepsilon$. 

We will show that even for completely general
systems, there exists a bound on the leading-order uncertainty, which can, in practice,
be even used to bound the proper trajectory uncertainty.

Next, we will assume stochastic model uncertainty by adding a white-noise-driven stochastic process in the perturbation
to the known vector field $f_{0}.$ This translates into a stochastic
differential equation of the form 
\begin{align}
\label{eq:startSDE}
    dx_{t}=f_{0}(x_{t},t)\,dt+\varepsilon g(x_{t},t,\varepsilon)dt+\varepsilon\sigma(x_{t,}t)\,\mathrm{d}W_{t},\\
    \quad x\in U\subset\mathbb{R}^{n},\quad t\in[t_{0},t_{1}],\quad0\leq\varepsilon\ll1. \nonumber
\end{align}
The SDE is understood in terms of the It\^{o} interpretation, where $W_t$ is an $n$-dimensional Wiener-process and $\sigma(x_t,t)$ is the covariance matrix.
The coefficient functions in \eqref{eq:startSDE} will be assumed to satisfy additional assumptions
that guarantee the existence of a strong solution to the equation. These types of stochastic
perturbations represent either random errors in the model or
unresolved effects. 

\section{\label{sec:deterministicSens}Deterministic Model sensitivity}

Traditionally, model sensitivity refers to the change in an observable under changes in the model equations, usually through parameters. 
Here, we take the observable to be simply the model trajectory itself.

\subsection{Influence of deterministic uncertainty}

Let us first assume that \eqref{eq:startingeq2} holds, generating a \emph{flow-map }
\begin{align}
F_{t_{0}}^{t} & :U\to U,\nonumber\\
x_{0} & \mapsto x(t,t_{0},x_{0};\varepsilon).
\end{align}

To assess the effect of a slight change in $\varepsilon$, we consider the norm of the difference between the idealized model trajectory
\begin{equation}
\label{eq:def1}
    x^0(t):= x(t,t_0,x_0;0)
\end{equation}
and the real one (with $\varepsilon\neq 0$) 
\begin{equation}
\label{eq:def2}
        x^\varepsilon(t):= x(t,t_0,x_0;\varepsilon)
\end{equation}
starting from the same initial condition. 
This \emph{trajectory uncertainty} is given by
\begin{equation}
    |x(t,t_0,x_0;0) - x(t,t_0,x_0;\varepsilon)|= |x^0(t) - x^\varepsilon(t)|.
\end{equation}

By classic results on ordinary differential equations, the flow map $F_{t_{0}}^{t}$
is as smooth in the parameter $\varepsilon$ as is the vector field $f_0 + \varepsilon g$,
and hence can also be Taylor-expanded in $\varepsilon$. This gives the \emph{leading-order trajectory uncertainty} as

\begin{equation}
 \varepsilon \left \vert \frac{\partial x^\varepsilon(t)}{\partial \varepsilon}\right \vert_{\varepsilon=0}  = \varepsilon |\eta(t,t_0,x_0)|.    
\end{equation}

The vector $\eta$, which is the derivative of the flow-map with respect to $\varepsilon$, obeys the (inhomogeneous) equation of variations\cite{Arnold1992}, also called the tangent model\cite{Cacuci} of \ref{eq:startingeq2}: 
\begin{align}
    \dot{\eta}&=\nabla f_{0}\left(x^0(t)\right)\eta+g\left(x^0(t),t;0\right), \nonumber\\
 \eta(t_{0};t_{0},x_{0})&=0.\label{eq:parameter eq of vari}
\end{align}
 The solution of this initial value problem is

\begin{equation}
    \eta(t;t_{0},x_{0})=\int_{t_{0}}^{t}\phi_{s}^{t}\left(x^0(s)\right)g\left(x^0(s),s;0\right)\,ds,
\end{equation}

where the deformation gradient, $\phi_{t_{0}}^{t}(x_{0})=\nabla F_{t_{0}}^{t}(x_{0})$,
is the normalized fundamental matrix solution of the equation of variations
\begin{equation}
\dot{\eta}=\nabla f_{0}\left(x^0(t),t\right)\eta,\label{eq:eq of variations}
\end{equation}
i.e., the homogeneous part of the linear system of ordinary differential
equations \eqref{eq:parameter eq of vari}. 

Therefore, the leading-order change to a trajectory $x^0(t)$
due to changes in the model is 

\begin{align}
\varepsilon\left|\eta(t;t_{0},x_{0})\right| & \nonumber =\varepsilon\left|\int_{t_{0}}^{t}\phi_{s}^{t}\left(x^0(s)\right)g\left(x^0(s),s;0\right)\,ds\right|\nonumber.\\
\end{align}

This quantity can be bounded from above as
\begin{align}
     \varepsilon\left|\eta(t;t_{0},x_{0})\right|& \leq\varepsilon\int_{t_{0}}^{t}\left|\phi_{s}^{t}\left(x^0(s)\right)g\left(x^0(s),s;0\right)\right|\,ds\nonumber \\
 & \leq\int_{t_{0}}^{t}\left\Vert \phi_{s}^{t}\left(x^0(s)\right)\right\Vert \left|\varepsilon g\left(x^0(s),s;0\right)\right|\,ds\nonumber \\
 & \leq\varepsilon\int_{t_{0}}^{t}\left\Vert \phi_{s}^{t}\left(x^0(s)\right)\right\Vert \,ds\left\Vert g\left(x^0(s),s;0\right)\right\Vert _{\infty}\nonumber \\
 & \leq\varepsilon\int_{t_{0}}^{t}\sqrt{\Lambda_{s}^{t}\left(x^0(s)\right)}\,ds\left\Vert g\left(x^0(s),s;0\right)\right\Vert _{\infty},\label{eq:eta estimate-1}
\end{align}
where $||\cdot||_\infty$ refers to the supremum norm and $\Lambda_{s}^{t}\left(x^0(s)\right)$ denotes the largest eigenvalue
of the (right) Cauchy\textendash Green strain tensor $C_{s}^{t}\left(x^0(s)\right)=\left[\phi_{s}^{t}\left(x^0(s)\right)\right]^{T}\phi_{s}^{t}\left(x^0(s)\right)$.
In other words, $\sqrt{\Lambda_{s}^{t}\left(x^0(s)\right)}$ is the largest singular
value of $\phi_{s}^{t}\left(x^0(s)\right)$. 

Let
\begin{align}
    \Delta_{\infty}(x_{0},t):&=\varepsilon\left\Vert g\left(x^0(\,\cdot\,),\,\cdot\,;0\right)\right\Vert _{\infty}= \\ \nonumber
    &=\varepsilon\max_{s\in\left[t_{0},t\right]}\left|g\left(x^0(s),s;0\right)\right|\label{eq:modelUncert}
\end{align}
\emph{denote the maximal leading-order model uncertainty }\cite{Kalnay2012} along the
trajectory $x^0(t)$ of the idealized model \eqref{eq:model}. With
this notation, let us define the \emph{leading-order trajectory uncertainty
}at any time instant $t\in[t_{0},t_{1}]$ as 

\begin{equation}
\delta(x_{0},t):=\varepsilon\left|\eta(t;t_{0},x_{0})\right|\leq\int_{t_{0}}^{t}\sqrt{\Lambda_{s}^{t}\left(x^0(s)\right)}\,ds\,\Delta_{\infty}(x_{0},t).\label{eq:sensitivitybound}
\end{equation}

For any finite $k\in\mathbb{N}^{+}$,
we also define the corresponding \emph{time averaged leading-order
trajectory uncertainty} as the temporal $L^{k}$ norm of $\delta(x_0,t)$:

\begin{equation}
    \delta_{k}(x_{0}):=\left\Vert \delta(x_0,t)\right\Vert _{L^{k}}=\varepsilon\sqrt[k]{\int_{t_{0}}^{t_{1}}\left[\eta(t;t_{0},x_{0})\right]^{k}dt}.
\end{equation}

 To obtain a uniform bound for $\delta_k(x_0)$ over the time interval $[t_0,t_1]$, we can simply
let $k\to \infty$ and find that

\begin{equation}
\delta_{\infty}(x_{0}):=\left\Vert\delta(x_0,t)\right\Vert _{\infty}=\varepsilon\max_{t\in\left[t_{0},t_{1}\right]}\left|\eta(t,t_{0},x_{0};0)\right|.
\end{equation}

Similarly, for the maximal leading-order model uncertainty, we can
set 
\begin{equation}
    \Delta_{\infty}(x_{0}):=\varepsilon\max_{s\in\left[t_{0},t\right]}\Delta_{\infty}(x_{0},s)=\varepsilon\max_{s\in\left[t_{0},t_{1}\right]}\left|g\left(x^0(s),s;0\right)\right|.
\end{equation}
Then, by Eq. \eqref{eq:eta estimate-1}, the trajectory uncertainty
$\delta_{k}(x_{0})$ can be estimated from above as
\begin{align}
\nonumber \delta_{k}(x_{0}) & =\varepsilon\sqrt[k]{\int_{t_{0}}^{t_{1}}\left[\eta(t;t_{0},x_{0})\right]^{k}dt}\\
 & \leq\Delta_{\infty}(x_{0},t)\left\Vert \int_{t_{0}}^{t}\sqrt{\Lambda_{s}^{t}\left(x^0(s)\right)}\,ds\right\Vert _{L^{k}}.
\end{align}

Taking the supremum norm of both sides gives 
\begin{equation}
\delta_{\infty}(x_{0})\leq\Delta_{\infty}(x_{0},t_{1})\max_{t\in\left[t_{0},t_{1}\right]}\int_{t_{0}}^{t}\sqrt{\Lambda_{s}^{t}\left(x^0(s)\right)}\,ds.\label{eq:L infinity estimate}
\end{equation}

Note  that, the upper bound $\int_{t_{0}}^{t}\sqrt{\Lambda_{s}^{t}\left(x^0(s)\right)}\,ds\,\Delta_{\infty}(x_{0},t)$
is generally not a monotone function of $t$. This implies that in order to
evaluate the right hand side of \eqref{eq:L infinity estimate}, one needs to compute the integral involved for all time instants in $[t_0,t_1]$.

{In conclusion, the estimate \eqref{eq:sensitivitybound} shows that the
leading-order trajectory uncertainty under modeling errors can be estimated from above
by a product of two quantities. One of these, $\Delta_\infty$, is a measure of the overall size of the model uncertainty, while the other factor is related to the sensitivity with respect to initial conditions \emph{within}
the idealized model. }
\begin{rem}
Recall that sensitivity with respect to initial conditions over a
time interval $[s,t]$ is typically characterized by the finite-time
Lyapunov exponent (or FTLE)\cite{Mathur2007}, given by
\begin{equation}
\label{eq:ftledef}
\mathrm{FTLE}{}_{s}^{t}\left(x^0(s)\right)=\frac{1}{t-s}\log\sqrt{\Lambda_{s}^{t}\left(x^0(s)\right)}.
\end{equation}

With this in mind, our uncertainty estimate can be rewritten as a
functional of the FTLE field as follows:
\begin{align}
\delta(x_{0},t) & \leq\Delta_{\infty}(x_{0},t)\int_{t_{0}}^{t}\exp\left[(t-s)\mathrm{FTLE}{}_{s}^{t}\left(x^0(s)\right)\right]\,ds.\label{eq:L 2 estimate-1}
\end{align}
\end{rem}

\begin{rem}
The calculation of the integral in \eqref{eq:sensitivitybound} can also be done in backward time, which is sometimes more convenient. Following the results on the smallest eigenvalue of the Cauchy-Green strain tensor\cite{Haller2011}, we note that
\begin{equation}
    \sqrt{\Lambda_{s}^{t}\left(x^0(s)\right)} = \frac{1}{\sqrt{\lambda_{\min}\left[C_{t}^{s}\left(x^0(t)\right)\right]}}
\end{equation}
and hence
\begin{equation}
\label{eq:backward}
    \int_{t_0}^t\sqrt{\Lambda_{s}^{t}\left(x^0(s)\right)}ds = \left \vert \int_{t}^{t_0}\frac{1}{\sqrt{\lambda_{\min}\left[C_{t}^{s}\left(x^0\right)\right]}}\right \vert.
\end{equation}

Numerically, formula \eqref{eq:backward} requires the evaluation of the integral of the square-root of the largest eigenvalue of backward-time Cauchy–Green strain tensor $C_{t}^{s}\left(x^0(t)\right)$, computed over $\left[t,s\right]$, with s decreasing from $t$ to $t_{0}$. This can be computed by finite-differencing along backward-time trajectories starting from a regular grid at time $t$, back to time $t_{0}$. The advantage of this approach is that the Cauchy-Green strain tensor is always calculated for the same initial point during integration. However, this point is the time-$t$ position of the idealized model trajectory. To obtain the bound as a function of the time-$t_0$ position, one needs to map the values back from time $t$ to time $t_0$ with the idealized flow-map. 

\end{rem}
With the above estimates, we can now bound the leading-order trajectory uncertainty
of the dynamical system. Most methods currently available for calculating sensitivity measures
need additional assumptions, such as the existence of an invariant measure\cite{Ruelle2009a}, ergodicity\cite{Wang2013b, Wang2014} or a specific form
of the modeling errors (to run direct simulations). While these often give
precise predictions on the value of the sensitivity, they are not
applicable to typical dynamical systems. In contrast,
the inequality \eqref{eq:sensitivitybound} holds for \emph{all} dynamical systems of the form \eqref{eq:startingeq2}. In addition, we can also use it to formulate
a bound on the proper (not only leading order) uncertainty in the dynamical system's trajectories.
\begin{thm}[Universal bound on trajectory uncertainty]
 Consider the dynamical system defined over a finite time interval
$[t_{0},t_{1}],$ and on a compact domain $U\subset\mathbb{R}^{n},$
by \eqref{eq:startingeq2}. Denote by $x^0(t)$
the idealized model's solution \eqref{eq:def1}, 
starting from $x_{0}$ at $t_{0}.$ Similarly, let $x^\varepsilon(t)$
be the true solution \eqref{eq:def2}, belonging to an arbitrary $\varepsilon\ne0$,
starting from the same initial condition.
Then, for any $\delta>0$ small enough, there exists $\varepsilon_{0}>0$,
such that for $\varepsilon<\varepsilon_{0}$ the following inequality holds for all $t\in[t_{0},t_{1}]$ and $x_0\in U$
\end{thm}
\begin{equation}
|x^\varepsilon(t)-x^{0}(t)|\leq\left(\int_{t_{0}}^{t}\sqrt{\Lambda_{s}^{t}\left(x^0(s)\right)}\,ds\,+\delta\right)\Delta_{\infty}(x_{0},t).\label{eq:uncertbound}
\end{equation}
\begin{proof}
See Appendix \ref{sec:app1}.
\end{proof}
Our Theorem 1 provides a bound for small values of $\varepsilon$,
that is computable numerically and holds for any time instant in
the time interval and any initial condition in the domain. At first,
the dependence on a finite $\delta$ may seem problematic. However, we note that for small enough $\varepsilon$, the size-$\delta$ correction can be made arbitrarily small. Our numerical findings indicate that \eqref{eq:uncertbound} tends to be satisfied even for $\delta = 0$. 

The inequality \eqref{eq:uncertbound} gives an upper bound for the
maximal possible error between the idealized model solution and the real
one. Available bounds in the literature\cite{Brauer1966,Kirchgraber1976} either require knowledge of the perturbed trajectory itself or introduce Gronwall-type estimates that vastly overestimate
the error, due to their universality in space and time.

For example, assume that in
system \eqref{eq:startingeq2}, $f_{0}$ satisfies the Lipschitz condition
with Lipschitz constant $L$ and the perturbation $\varepsilon g(x,t)$
is uniformly bounded by a constant $M=\max\Delta_{\infty}.$ We then
obtain $|x^\varepsilon(t)-x^{0}(t)|\leq\frac{M}{L}(e^{L(t-t_{0})}-1)$ from the classic Gronwall-lemma \cite{GuckenheimerHolmes}. This
is a rigorous but highly conservative upper bound on the trajectory uncertainty, as seen from a direct comparison with \eqref{eq:uncertbound}.

To illustrate the difference between the two estimates, consider the classic damped-forced Duffing-oscillator
\begin{align}
\dot{x} & =y,\label{eq:duffing}\\
\dot{y} & =x-x^{3}-\delta y+A\cos t, \nonumber 
\end{align}

with $\delta=0.15$ and $A=0.3$. {For these parameter values, the system is chaotic. As such, it is reasonable to expect high sensitivity to modeling errors. It is also known that the system has a chaotic global attractor\cite{Hadjighasem2013}, contained in the region $U=[-1.5,1.5]\times[-1.5,1.5]$, which is an invariant set of the stroboscopic map\cite{GuckenheimerHolmes} of period $2\pi$}. 
Using this fact, we choose the global Lipschitz-constant $L=2$ over this bounded domain for the idealized system.
We consider a model error term of the form $\varepsilon g(x,t)=(0,\varepsilon\sin(\omega_{p}t))$,
representing a high-frequency deterministic perturbation
to system \eqref{eq:duffing}. For this choice of perturbation,
we can select the uniform bound $M=\varepsilon$ for the model error.

\begin{figure}[h]
\includegraphics[width=0.49\textwidth]{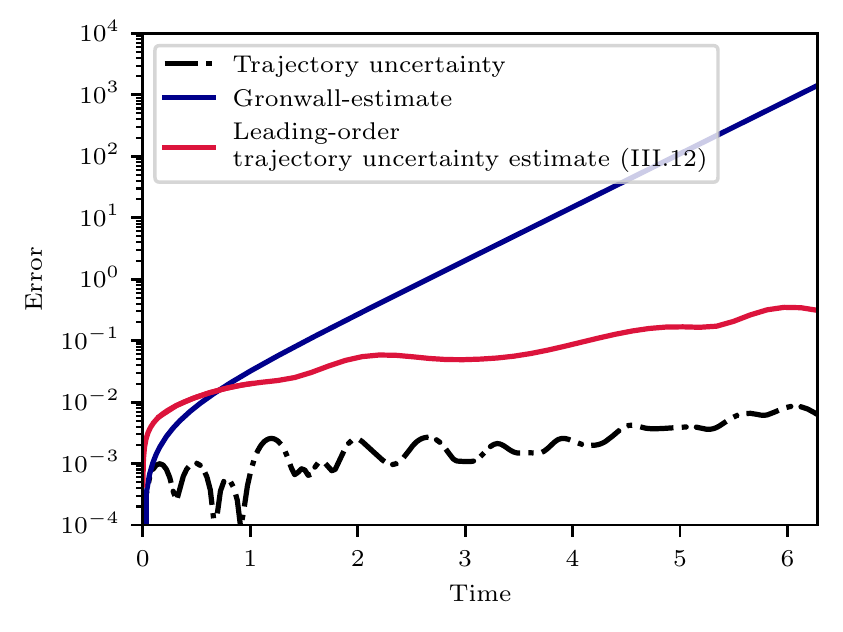}\caption{\label{fig:fig1}Comparison of uncertainty estimates, applied to system
\ref{eq:duffing}. The dashed-dotted curve shows the phase-space distance
between the idealized model solution and the real
one, started from the same
initial condition $(0.15,0.4)$, with $\varepsilon=0.01$ and $\omega_p=10$. The red
curve is the bound obtained from the leading-order bound in inequality
\eqref{eq:sensitivitybound}. The blue curve is the Gronwall-type, rigorous
upper bound, defined by $\frac{M}{L}(e^{Lt}-1)$, with $M=\varepsilon$
and $L=2$. }
\end{figure}

As seen in Fig \ref{fig:fig1}, both the Gronwall-type estimate and the leading-order bound of \eqref{eq:sensitivitybound} substantially overestimate 
the actual distance between the true and the idealized model trajectories.
However, while the Gronwall-estimate suggests an overall exponential increase for all trajectories starting in $U$, our leading-order bound \eqref{eq:leadbound} depends on the unperturbed trajectory, providing a tighter estimate on the trajectory uncertainty. 

\section{\noun{\label{sec:Stochastic-error}Stochastic model sensitivity}}

It is often reasonable to assume a stochastic
model error as one of the sources of uncertainty in the model. In that case, trajectories obey the following stochastic differential equation (SDE) 
\begin{equation}
    \label{eq:whitenoise}
    \dot{x}=f_0(x,t)+\varepsilon g(x,t,\varepsilon) + \varepsilon\sigma(x,t)\xi(t).
\end{equation}

The uncertainty comes from a white-noise process $\xi$ and $f_0,g, \sigma$
are smooth functions. Equation \eqref{eq:whitenoise} can be interpreted in the It{\^o}-sense
as
\begin{equation}
dx_{t}=f_0(x_{t},t)\,dt + \varepsilon g(x_t, t, \varepsilon)\, dt+\varepsilon\sigma(x_{t},t)\,dW_{t},\label{eq:SDE}
\end{equation}
on a probability space $(\Omega,\mathcal{F},P)$, with $W_{t}$ being
an $n$-dimensional Wiener process, $f:\mathbb{R}^{n}\times[t_{0},t_{1}]\times\mathbb{R}\to\mathbb{R}^{n}$
is the deterministic part of the SDE, and $\sigma(\cdot,\cdot):\mathbb{R}^{n}\times[t_{0},t_{1}]\to\mathbb{R}^{n\times n}$
is the covariance matrix of the noise. Both $f$ and $\sigma$ are
assumed to be measurable, smooth functions of their arguments.

Our goal is to characterize the leading-order deviation of the solution process
$x_{t}$ of (\ref{eq:SDE}) from the solution of the idealized model ($\varepsilon=0$). Note that the idealized model dynamics is given by the ODE (\ref{eq:model}),
for all realizations $\omega\in\Omega$ of the noise. To achieve such
a characterization, we develop an upper estimate similar to (\ref{eq:L infinity estimate}).
We first state the necessary and sufficient conditions for the existence
of a solution process $x_{t}$, derive the
SDE governing the leading-order trajectory uncertainty
(a stochastic analog to the equation of variations), and give bounds
on the expected value of the norm of its solutions. 

Assume that there exist constants $ C,D  >0$, such that for all $x,y\in\mathbb{R}^{n}, \ t\in[t_0,t_1]$ and small enough $\varepsilon>0$, we have
\begin{align}
  |f_0(x,t) + \varepsilon g(x,t,\varepsilon)|+|\varepsilon\sigma(x,t)|&\leq C(1+|x|),\nonumber\\
  |f_0(x,t) + \varepsilon g(x,t,\varepsilon)-f_0(y,t) - \varepsilon g(y,t,\varepsilon)| \nonumber\\
 +|\varepsilon\sigma(x,t)-\varepsilon\sigma(y,t)|&\leq D|x-y|.
\end{align}

Then, Equation (\ref{eq:SDE}) along with the deterministic initial
condition $x_{t=t_{0}}=x_{0}$ has a unique solution $x_{t}$ which
is adapted to the filtration generated by $W_{s}$ for $s\leq t$.
In addition, $E\left(\int_{t_{0}}^{t_{1}}|x_{t}|^{2}dt\right)<\infty$
holds and the sample paths of the solution $x_{t}(\omega)$ are continuous \cite{Oksendal2010}. The following theorem provides an analogue of the equation of variations \eqref{eq:eq of variations} in the stochastic setting.
\begin{thm}[Small noise expansion]
Assume that the coefficients in (\ref{eq:SDE}) have bounded and measurable
partial derivatives up to second order. Then, there exists $\bar{\varepsilon}>0$, such that for $\varepsilon<\bar{\varepsilon}$ the solution $x^\varepsilon_{t}$ can
be written as 
\begin{equation}
\label{eq:smallnoise}
    x^\varepsilon_{t}=x^0_t+\varepsilon\eta_{t}+ \varepsilon^2 R_{2}(t,\varepsilon),
\end{equation}
with the same notation as we had in (\ref{eq:parameter eq of vari}), but now with $\eta_{t}$ denoting a stochastic process. The
remainder term, $R_{2}(t,\varepsilon)$, is bounded in the mean-squared sense, i.e., there exists
$K>0$, such that
\begin{equation}
    \sup_{t\in[t_0,t_1]}\left[ \EX \left\vert R_2(t,\varepsilon)\right \vert^2\right]\leq K.
\end{equation}

The coefficients $x^0_t$ and $\eta_{t}$ satisfy the
system of stochastic differential equations
\begin{align}
\label{eq:smallnoise1}
dx^0_t =&f_{0}\left(x^0_t,t\right)dt,\quad x^0_{t=t_0}=x_{0},\\ 
d\eta_{t} =&\nabla f_{0}\left(x^0_t,t\right)\eta_{t}dt\nonumber\\
 & +g\left(x^0_t,t;0\right)dt+\sigma\left(x^0_t,t\right)dW_{t},\nonumber\\ 
  \eta_{t=t_{0}}=&0\label{eq:smallnoise2}. 
\end{align}
\end{thm}
\begin{proof}
This result is the application of the \emph{small-noise expansion
}of stochastic differential equations \citep{blago1962,Friedlin2012,gardiner2004},
which is analogous to the equation of variations for
ordinary differential equations. The proof is essentially the extension
of the known result for the vector-valued, autonomous case \citep{Albeverio2015}, to also allow for nonautonomous and parameter-dependent SDE-s.
For details, see Appendix \ref{sec:app2}. 
\end{proof}
\begin{rem}
The zeroth-order SDE in $\varepsilon$, Eq. \eqref{eq:smallnoise1}, is precisely the idealized model. Hence, the solution process $x^0_t$ is deterministic and could be also written as $x^0_t \equiv x^0(t)$.
\end{rem}
\begin{thm}
Let $\phi_{t_{0}}^{t}(x_{0})$ be the normalized fundamental
matrix solution to (\ref{eq:eq of variations}). Then, $\eta_{t}$
defined as the solution to the linear SDE \eqref{eq:smallnoise2}, is an Ornstein-Uhlenbeck
process that can be written as
\begin{equation}
    \eta_{t}=\int_{t_{0}}^{t}\phi_{s}^{t}\left(x^0_s\right)g \left(x^0_s, s\right)\,ds+\int_{t_{0}}^{t}\phi_{s}^{t}\left(x^0_{s}\right)\sigma\left(x^0_s,s\right)\,dW_{s}.
    \label{eq:eta_stoch}
\end{equation}
\end{thm}
\begin{proof}
This result is well-known for scalar stochastic differential equations. The extension to our multi-dimensional setting is given in Appendix \ref{sec:app3}. 
\end{proof}

Following this result, let $N(t)=||\eta_{t}||$ denote the norm of
the vector valued stochastic process $\eta_{t}$, which measures the
leading-order trajectory uncertainty arising from \emph{both deterministic and stochastic }modeling errors. The leading-order trajectory uncertainty
is then $\varepsilon N(t)$.  Using formula \eqref{eq:eta_stoch} for $\eta_{t}$, we can define the deterministic
term ($N_{d}$), the stochastic term $(N_{s})$ and the mixed term
($N_{m}$) of this leading-order trajectory uncertainty. To remain
consistent with the notation of Section \ref{sec:deterministicSens},
we have the deterministic term $\delta(x_{0},t)=\varepsilon N_{d}(t)$. The full expression for $N^2(t)$ is:
\begin{align}
 \label{eq:nsquared}
N(t)&^{2} \\
=&\left(\int_{t_{0}}^{t}\phi_{s}^{t}\left(x^0_{s}\right)g\left(x^0_s, s\right)\,ds+\int_{t_{0}}^{t}\phi_{s}^{t}\left(x^0_{s}\right)\sigma\left(x^0_{s},s\right)\,dW_{s}\right)^{2}\nonumber\\
 =&\left(\int_{t_{0}}^{t}\phi_{s}^{t}\left(x^0_{s}\right)g\left(x^0_s,s\right)\,ds\right)^{2}+\left(\int_{t_{0}}^{t}\phi_{s}^{t}\left(x^0_{s}\right)\sigma\left(x^0_{s},s\right)\,dW_{s}\right)^{2}\nonumber \\
 & +2\left(\int_{t_{0}}^{t}\phi_{s}^{t}\left(x^0_{s}\right)\sigma\left(x^0_{s},s\right)\,dW_{s}\right)\left(\int_{t_{0}}^{t}\phi_{s}^{t}\left(x^0_{s}\right)g\left(x^0_s,s\right)\,ds\right)\nonumber \\
 =&N_{d}(t)^{2}+N_{s}(t)^{2}+2N_{m}(t)\nonumber 
\end{align}

Formula \eqref{eq:nsquared} allows us to formulate a stochastic extension of Theorem 1, which applies
even in the stochastic setting. The quantity to be estimated is now the mean-square of the leading-order trajectory uncertainty. 
\begin{thm}[Bound on the mean-squared leading-order trajectory uncertainty]
The leading-order trajectory uncertainty can be bounded
in the mean-square sense as
\begin{align}
\label{eq:SDEbound}
\varepsilon^2 \EX\left[N(t)^{2}\right]\leq&\left(\int_{t_{0}}^{t}\sqrt{\Lambda_{s}^{t}\left(x^0_{s}\right)}\,ds\right)^{2}\,\Delta_{\infty}^{2}(x_{0},t) \\&+\int_{t_{0}}^{t}\mathrm{tr}\left[C_{s}^{t}\left(x^0_{s}\right)\right]\,ds\,\Delta_{\infty}^{\sigma}(x_{0},t),\nonumber     
\end{align}
where we have introduced the notation $\Delta_{\infty}^{\sigma}(x_{0},t)=\varepsilon^{2}\max_{s\in\left[t_{0},t\right]}\mathrm{tr}\left[\sigma\left(x^0_{s},s\right)^{T}\sigma\left(x^0_{s},s\right)\right]$.
\end{thm}

\begin{proof}
The proof consists of a computation of the expected values of $N_{s}^{2}$
and $N_{m},$ since $N_{d}^{2}$ is purely deterministic and was already
computed before. The details of the proof are given in Appendix \ref{sec:app4}. 
\end{proof}

Note that if the model has no stochastic
error, i.e., $\sigma(x,t)\equiv0,$ Theorem 4 gives $N(t)=N_{d}$
and $\Delta_{\infty}^{\sigma}(x_{0},t)\equiv0$, yielding
the upper estimate $\varepsilon \EX \left[N(t)\right]=\varepsilon N(t)=\delta(x_{0},t)\leq\int_{t_{0}}^{t}\sqrt{\Lambda_{s}^{t}\left(x^0_{s}\right)}\,ds\,\Delta_{\infty}(x_{0},t)$.
This is consistent with the upper bound derived in Section \ref{sec:deterministicSens}.

Rearranging expression \eqref{eq:SDEbound}, we obtain a quantity, computed in terms of the \emph{idealized} model and the relative strength of errors (deterministic or stochastic). We refer to this quantity as \emph{Model Sensitivity (MS)}, defined as
\begin{equation}
\label{eq:MSdef}
    \text{MS}_{t_0}^t(x_0;r) := \left( \int_{t_0}^t \sqrt{\Lambda_s^t\left(x^0_s\right)}ds\right)^2 + r \int_{t_0}^t \text{tr}[C_s^t\left(x^0_s\right)]ds,
\end{equation}
where $r:= \Delta^\sigma_\infty(x_0,t)/\Delta^2_\infty(x_0,t)$ is the ratio characterizing the relative importance of the stochastic modeling errors. By calculating MS$_{t_0}^t$ for several initial conditions in a phase-space region of interest, we can quickly identify locations of high sensitivity to modeling errors. By Theorem 4, these locations are expected to show higher uncertainty. 

{We note that MS is a scalar-valued function of several variables: it depends on the phase-space location and the chosen time interval. Therefore, it does not give a {\em global} characterization of the model's sensitivity. Instead, we must view it as a time-dependent scalar field, which provides granular analysis of sensitivities. A similar assessment of sensitivities distributed over phase space was recently given by using Markov modeling in the context of response theory\cite{gutierrez2020}. That result focuses on infinite time intervals, which is not the case for our method.}

{Moreover, by Theorem 4, the leading order trajectory uncertainty is related to MS, in the mean-square sense, by
\begin{align}
    \label{eq:leadbound}
    \varepsilon^2\EX[N(t)^2] &\leq \text{MS}_{t_0}^t(x_0;r) \Delta^2_\infty(x_0,t), \text{ or equivalently, } \nonumber\\
    \text{MS}_{t_0}^t(x_0;r)&\geq \frac{\varepsilon^2\EX[N(t)^2]}{\Delta^2_\infty(x_0,t)}.
\end{align}
In other words, MS is the coefficient relating the leading order mean-squared trajectory uncertainty to the modeling uncertainty. }

As in the purely deterministic case, we obtain a theorem that relates $(\text{MS}_{t_0}^t)$ to the proper trajectory uncertainty. 

\begin{thm}[Universal bound on the mean-squared trajectory uncertainty]
Consider the stochastic dynamical system defined over a finite time interval $[t_0, t_1]$ and on a compact domain $U\subset \mathbb{R}^n$, by  the SDE \eqref{eq:SDE}. Then, for any $\delta>0$ there exists an $\varepsilon_0>0$, such that for $\varepsilon <\varepsilon_0$ the following inequality holds for all $t\in [t_0,t_1]$ and $x_0\in U$:
\begin{equation}
    \label{eq:theorem5}
    \sqrt{\EX\left (|x^\varepsilon_t - x^{0}(t)|^2\right)} \leq  \Delta_\infty(x_0,t)\left(\sqrt{\text{MS}_{t_0}^t(x_0,r)}+\delta\right).
\end{equation}
\end{thm}
\begin{proof}
See Appendix \ref{sec:app5}.
\end{proof}
By Theorem 5, the bound on the mean-squared leading-order trajectory uncertainty is extended to the actual mean-squared trajectory uncertainty, for small enough $\varepsilon$. Then, the MS can be used to calculate a time-dependent upper bound on the trajectory uncertainty, which will be true for any perturbation of size $\Delta_\infty$, assuming a ratio of $r$ between stochastic and deterministic modeling errors.

We also note that in practice, the bound \eqref{eq:theorem5} tends to be satisfied even without including the size-$\delta$ correction (similarly to Theorem 1). This means that the much simpler expression of Theorem 4 can be used to assess the mean-squared trajectory uncertainty. In the next section, we demonstrate this fact on a few examples.  
\section{\label{sec:computation}Computation of trajectory uncertainty estimates}

We start by an explicit calculation of MS for linear systems. Within this class of systems, we can find examples proving the optimality of our estimates. 
Consider the constant coefficient linear
stochastic differential equation, driven by an $n-$dimensional Wiener-process
$\mathbf{W}_{t}$, 

\begin{align}
d\mathbf{x}_{t}  =\mathbf{A}\mathbf{x}_{t}\,dt+\varepsilon \mathbf{b}dt+\varepsilon\sigma\,d\mathbf{W}_{t},\qquad \mathbf{x},\mathbf{b}\in\mathbb{R}^{n},\quad  \mathbf{A}, \mathbf{\sigma}\in\mathbb{R}^{n\times n}.
\label{eq:linsde}
\end{align}

Here, $\mathbf{b}$ is a (constant) deterministic perturbation vector, $\mathbf{\sigma}$ is
the covariance matrix of the noise and $\varepsilon \geq 0$ controls the
size of the perturbation. 

To calculate the MS, we use formula \eqref{eq:MSdef}, with
\begin{equation}
    \Delta_\infty = \varepsilon|\mathbf{b}| \qquad \Delta^\sigma_\infty=\varepsilon^2 ||\sigma||^2_F.
\end{equation}
The equation of variations of system \eqref{eq:linsde} is simply $\dot{\phi}_{t_0}^t = \mathbf{A}\phi_{t_0}^t$, which gives $\phi^t_{t_0} = e^{\mathbf{A}(t-t_0)}$ for the flow-map gradient.
Then, by formula \eqref{eq:MSdef}, MS is
\begin{align}
\text{MS}_{t_0}^t = &\left(\int_{t_{0}}^{t}\sqrt{\Lambda\left[\left(e^{\mathbf{A}(t-s)}\right)^Te^{\mathbf{A}(t-s)}\right]}\,ds\right)^{2}\nonumber \\
&+\frac{||\sigma ||^2_F}{|\mathbf{b}|^2}\int_{t_{0}}^{t}\mathrm{tr}\left[\left(e^{\mathbf{A}(t-s)}\right)^T e^{\mathbf{A}(t-s)}\right]\,ds.     
\end{align}

From this, we can obtain the bound on the leading-order trajectory uncertainty after multiplying by $\varepsilon^2 |\mathbf{b}|^2$.

On the other hand, we can calculate the trajectory uncertainty directly.
The idealized system
(with $\varepsilon=0$) has the general solution $\mathbf{x}_{t}^{0}=e^{\mathbf{A}t}\mathbf{x}_{0}$,
while the solution to the perturbed problem is the stochastic process\cite{Oksendal2010}

\begin{equation}
    \mathbf{x}_{t}=e^{\mathbf{A}t}\mathbf{x}_{0}+\varepsilon \int_{t_{0}}^{t}e^{\mathbf{A}(t-s)}\mathbf{b} ds+\varepsilon\int_{t_{0}}^{t}e^{\mathbf{A}(t-s)}\sigma d\mathbf{W}_{s}.
\end{equation}

The mean-square of the difference between the idealized model solution, and the real solution is 
\begin{align}
    &\EX(|\mathbf{x}^\varepsilon_{t}-\mathbf{x}_{t}^{0}|^{2}) \nonumber\\
    &=\varepsilon^2\left( \int_{t_{0}}^{t}e^{\mathbf{A}(t-s)}\mathbf{b} ds\right)^2 + \varepsilon^2\EX \left(\int_{t_{0}}^{t}e^{\mathbf{A}(t-s)}\sigma d\mathbf{W}_{s}\right)^2 \nonumber\\
    &= \varepsilon^2\left( \int_{t_{0}}^{t}e^{\mathbf{A}(t-s)}\mathbf{b} ds\right)^2 + \varepsilon^2\int_{t_{0}}^{t}||e^{\mathbf{A}(t-s)}\sigma||^2_F ds.
\end{align}
Here, we used that the expected value of the mixed term is zero and the expression for the second integral follows from It{\^o}'s isometry.  

An immediate consequence of this calculation is the optimality of Theorem 4.
If system \eqref{eq:linsde} is a scalar equation, $x_t \in \mathbb{R}, A, \sigma, b \in \mathbb{R}$, then once we evaluate the integrals, we obtain
\begin{align}
\EX(|x^\varepsilon_{t}-x_{t}^{0}|^{2})&= \frac{\varepsilon^2 b^2}{A^2}\left(e^{A(t-t_0)} - 1\right)^2 + \frac{\varepsilon^2\sigma^2}{2 A}(e^{2A(t-t_0)}-1)\nonumber\\    
&= \varepsilon^2 b^2 \text{MS}_{t_0}^t.
\end{align}

This shows that Theorem 4 is optimal: the bound it provides cannot be strengthened for general systems.
\subsection{Numerical examples}
\begin{example} The Duffing oscillator \end{example}
{To illustrate our main results, we apply
formula (\ref{eq:leadbound}) to two models of differing complexity.
First, let us consider once again the damped-driven Duffing oscillator, defined by \eqref{eq:duffing}, which exhibits chaotic behavior}. In the presence of a deterministic, time-periodic perturbation, the trajectory uncertainty was already shown in Fig. \ref{fig:fig1}.  To assess the sensitivity to general, possibly stochastic perturbations, we first calculate MS$_{t_0}^t$ and display it
on a uniform grid over the domain $U=[-1.5,1.5]\times[-1.5,1.5],$
for two time intervals of interest, $[0,2\pi]$ and $[0,4\pi]$. This calculation only requires knowledge of the idealized system and the relative magnitude of modeling errors. For the calculation of the Cauchy-Green strain tensor, we use finite differences, over a secondary grid \cite{Onu2015} to increase accuracy.

\begin{figure}[h]
\includegraphics[width=0.49\textwidth]{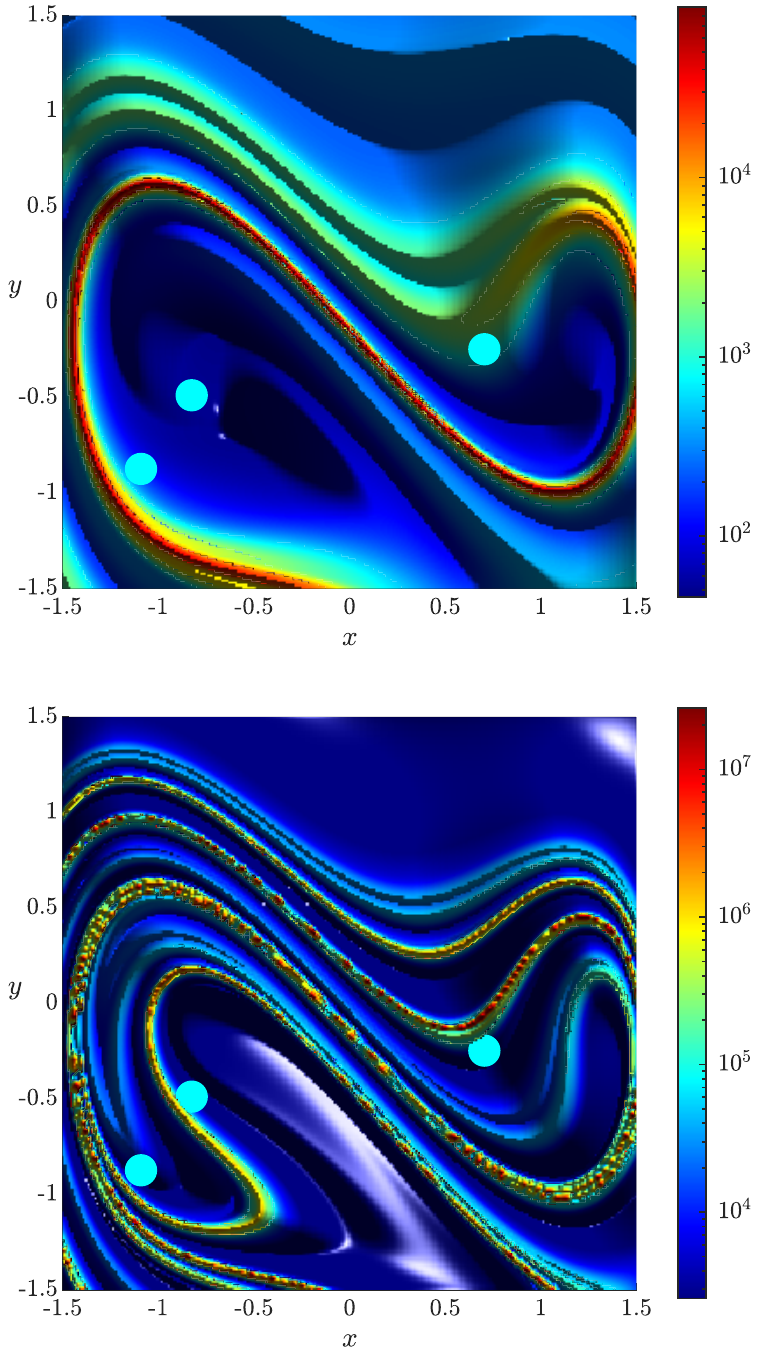}
\caption{Model sensitivity (MS) for the Duffing oscillator under both deterministic and stochastic modeling errors. The value of the MS is obtained from formula \eqref{eq:MSdef}, applied to
system \eqref{eq:duffing} with parameters $\delta=0.15,$ $A=0.3$.
Both the deterministic model error and the noise is assumed to have
amplitude $\varepsilon$, that is $\Delta_{\infty}(x,t)=\varepsilon$,
$\Delta_{\infty}^{\sigma}(x,t)=\varepsilon^{2}$, with $\varepsilon=0.01$. In the upper panel, the time interval of interest is $[0, 2\pi]$, while in the lower panel, it is $[0, 4\pi]$. Light blue dots mark the starting points of the trajectories relevant for Fig. \ref{fig:duffingmc}.}\label{fig:duffing1}
\end{figure}

MS fields are shown in Fig. \ref{fig:duffing1}.
We now assume a specific modeling error that contains both a deterministic and a stochastic component. The equations then are SDEs, which read as
\begin{align}
    dx_t &= ydt, \label{eq:sdeDuff}\\
    dy_t &= (x_t - x_t^3 -\delta y_t + A \cos t)dt + \varepsilon \sin(\omega_p t)dt + \varepsilon dW_t. \nonumber
\end{align}
In this case, both types of errors are assumed to be of norm $\varepsilon$, that is, $\Delta_\infty^2 = \Delta_\infty^\sigma = \varepsilon^2$, with $\omega_p=10$.

We compare the bound \eqref{eq:leadbound} on the leading-order trajectory uncertainty, obtained from MS, with the actual observed mean-squared trajectory uncertainty at select initial conditions. We calculate the mean-squared trajectory uncertainty from 2000 realizations of the stochastic process defined by \eqref{eq:sdeDuff}. For the solution of the SDE, an Euler-Maruyama scheme is used. The phase-space locations of the initial conditions considered are marked in Fig. \ref{fig:duffing1}. 

\begin{figure*}
    \centering
    \includegraphics[width=\textwidth]{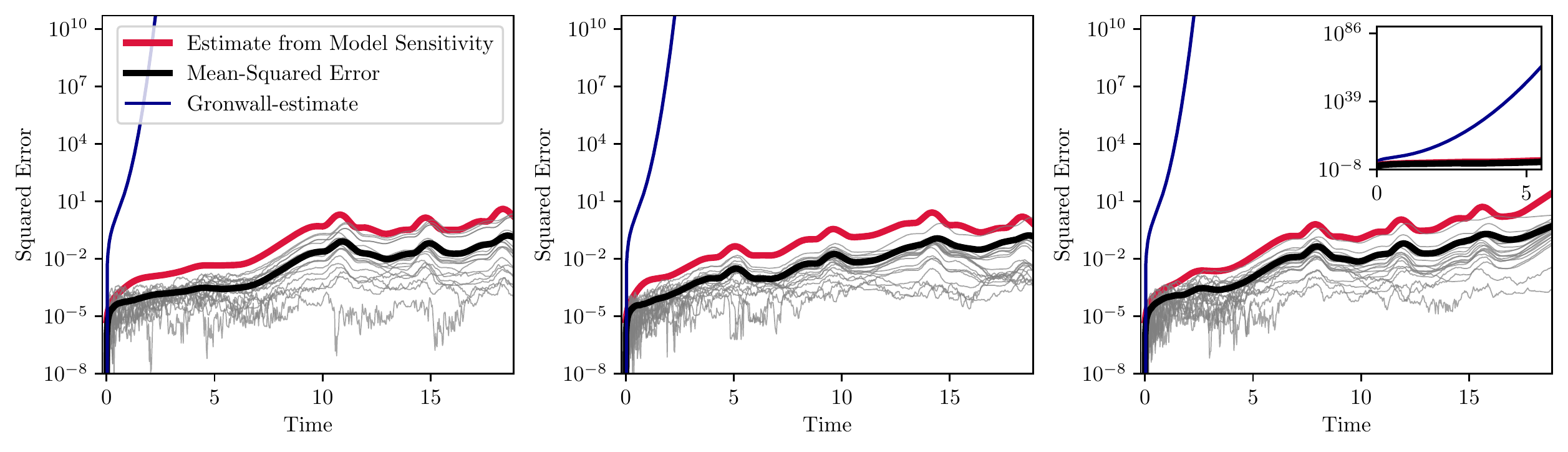}
    \caption{Square of the difference between the idealized model solutions and the perturbed solutions to the Duffing-system \eqref{eq:sdeDuff}. The grey curves show the error along a few sample paths, the black curve is the mean-squared error computed from 2000 sample paths. The red curve is the upper bound on the mean-squared leading-order trajectory uncertainty, defined by MS$_{t_0}^t(x_0,r)\Delta_\infty^2(x_0,t)$. The blue curve is the Gronwall-type upper bound for the mean-squared error\cite{note}, asymptotically given as $e^{L^2t^2/2}$, where $L$ is a Lipschitz-constant for \eqref{eq:duffing} and was chosen to be $L=2$. The inset shows the three curves on a larger scale. The modeling errors are detailed in the text, $\Delta_\infty^2=\varepsilon^2$ with $\varepsilon=0.01$ and $r=1$. The initial conditions are: left panel: $x_0 = (-0.8253, -0.48795)$, middle panel: $x_0 = (-1.0904,-0.87349)$, right panel: $x_0=(0.70482, -0.24699)$.}\label{fig:duffingmc}
\end{figure*}

The estimated upper bounds on the trajectory uncertainties are shown in the three panels of Fig. \ref{fig:duffingmc}. For all three initial conditions, we see that the bound on the expected mean-squared trajectory uncertainty is confirmed.
{Remarkably, Fig. \ref{fig:duffingmc} shows that the mean-squared trajectory uncertainty stays within two orders of magnitude of the leading-order bound, closely following trends in its graph. In other words, not only is MS a quantitatively accurate upper estimate, but it also provides qualitative information about the time dependence of the error growth.

For this particular system, the upper bound has predictive power over finite time intervals. The reason is that the idealized model has an underlying chaotic attractor (of finite size) and thus does not allow unbounded growth of errors. This is not the case, for example, in system \eqref{eq:linsde} (with $n=1$), where the bound was shown to be attained exactly for all times, providing an infinitely large relevant time interval. However, we also note that in practice, this relevant time interval can be quite long, much longer than what we would consider relevant for a rigorous, Gronwall-type estimate, which can also be derived for stochastic modeling errors\cite{note}. See Fig. \ref{fig:fig1} for a comparison in the purely deterministic case or Fig. \ref{fig:duffingmc} in the stochastic case. }

As noted earlier for the calculation of MS$_{t_0}^{t}(x_0, r)$, as well as for the leading-order trajectory uncertainty, we did not make any assumptions on the \emph{form} of the modeling errors. For this reason, given one specific instance of modeling errors and two points $x^{(1)}_0$ and $x^{(2)}_0$, the relation
\begin{equation}
    \label{eq:ms1ms2}
    \text{MS}_{t_0}^t\left(x^{(1)}_0,r\right) <     \text{MS}_{t_0}^t\left(x^{(2)}_0,r\right)
\end{equation}
does not imply that the actual trajectory uncertainty will be greater in $x^{(2)}_0$ than in $x^{(1)}_0$. 
Instead, what we can conclude from \eqref{eq:ms1ms2}, is that the dynamics at $x^{(1)}_0$ is such, that it can allow higher trajectory uncertainty than at $x^{(2)}_0$.

\begin{example}The Charney- DeVore model\end{example}

Next, we turn to a higher-dimensional model\cite{Charney1979}. It is demonstrated in
Ref. \onlinecite{Crommelin2004}, that a six dimensional reduced order model
for barotropic flow over topography admits multiple
equilibria, and can even exhibit tipping transitions between them.
Therefore, the \emph{Charney-DeVore}\cite{Charney1979} model is expected to show
highly unstable transient behavior \citep{Babaee2017}, which results
in high sensitivity with respect to perturbations. The dynamical equations
are

\begin{align}
\dot{x}_{1} & =\widetilde{\gamma}_{1}x_{3}-C(x_{1}-x_{1}^{*}),\nonumber \\
\dot{x}_{2} & =-(\alpha_{1}x_{1}-\beta_{1})x_{3}-Cx_{2}-\delta_{1}x_{4}x_{6},\nonumber \\
\dot{x}_{3} & =(\alpha_{1}x_{1}-\beta_{1})x_{2}-\gamma_{1}x_{1}-Cx_{3}+\delta_{1}x_{4}x_{5},\nonumber \\
\dot{x}_{4} & =\widetilde{\gamma}_{2}x_{6}-C(x_{4}-x_{4}^{*})+\lambda(x_{2}x_{6}-x_{3}x_{5}),\nonumber \\
\dot{x}_{5} & =-(\alpha_{2}x_{1}-\beta_{2})x_{6}-Cx_{5}-\delta_{2}x_{4}x_{3},\label{eq:cdv}\\
\dot{x}_{6} & =(\alpha_{2}x_{1}-\beta_{2})x_{5}-\gamma_{2}x_{4}-Cx_{6}+\delta_{2}x_{4}x_{2}.\nonumber 
\end{align}

The coefficients $\alpha_m$, $\beta_m$, $\gamma_m$, $\delta_m$ are defined by
\begin{align}
    &\alpha_m = \frac{8 \sqrt{2}}{\pi}\frac{m^2}{4m^2-1}\frac{b^2 + m^2 -1}{b^2+m^2}, \qquad \beta_m=\frac{\beta b^2}{b^2+m^2}, \nonumber\\
    &\delta_m = \frac{64\sqrt{2}}{15\pi}\frac{b^2-m^2+1}{b^2+m^2}, \qquad \widetilde{\gamma}_m=\gamma \frac{4m}{4m^2-1}\frac{\sqrt{2}b}{\pi}, \nonumber\\
    &\lambda = \frac{16 \sqrt{2}}{5\pi}, \qquad \gamma_m = \gamma \frac{4m^3}{4m^2-1}\frac{\sqrt{2}b}{\pi(b^2+m^2)}.
\end{align}
As in Refs. \onlinecite{Babaee2017,Crommelin2004}, we set the parameters, to correspond to the multistable regime: $(x_1^*, x_4^*, C, \beta, \gamma, b)= (0.95, -0.76095, 0.1, 1.25, 0.2, 0.5)$. 

\begin{figure*}
\includegraphics[width = \textwidth]{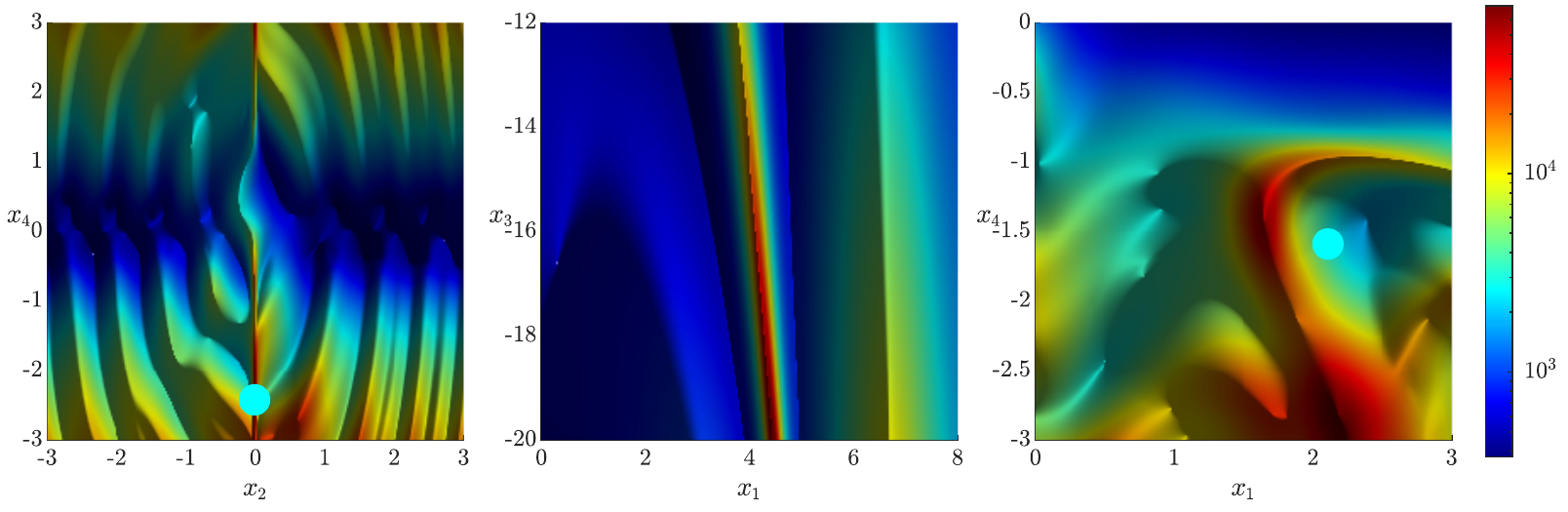}
\caption{Model sensitivity (MS), computed for the Charney-DeVore
model \eqref{eq:cdv}. The parameter values used are given in the
text, the time interval is $t_{0}=0,$ $t=15.$ The ratio of the importance of stochastic and deterministic modeling errors was set to $r=1$. The figures show different
two-dimensional slices of the six-dimensional phase space. Light blue dots mark the starting points of the trajectories relevant for Fig. \ref{fig:cdvmc}.}
\label{fig:cdv1}
\end{figure*}

The MS$_{t_0}^t$ field is shown in Fig. \ref{fig:cdv1} along a few slices of phase space. Similarly to the low-dimensional Duffing oscillator, the phase space of the Charney-DeVore model also exhibits high variability for MS. 

Although we cannot show the complete MS field in this high-dimensional phase space, this example demonstrates how our method remains applicable in higher-dimensional systems. Even in this lower-dimensional representation, we can distinguish structures, with particularly high sensitivity to perturbations over the chosen time scale. 

Next, we fix a modeling error to the equations \eqref{eq:cdv} in the form
\begin{equation}
    \mathbf{g}(\mathbf{x},t) = \mathbf{b_0}\sin(\mathbf{k}\cdot \mathbf{x})\cos(\omega_p t),\quad |\mathbf{b_0}|=1,\quad \mathbf{\sigma} = \frac{1}{\sqrt{6}}\mathbb{I}.
\end{equation}
This represents a deterministic modeling error that is periodic in both time and space. In addition, each coordinate is perturbed by an independent Wiener-process. 

The vector $\mathbf{b}_0$ is of unit length and has components $\mathbf{b}_0 = (0.310, 0.376, 0.476, 0.478, 0.281, 0.479)$. The wave vector is $\mathbf{k}=(1.815, 1.905, 1.127, 1.913, 1.632, 1.097)$, $\omega_p =10$. With this choice of the parameters, the magnitude of the perturbations is once again $\Delta_\infty^2(x_0,t)=\varepsilon^2|\mathbf{b}_0|^2 =\varepsilon^2$ and $\Delta_\infty^\sigma= \varepsilon^2 \text{tr }\mathbf{\sigma}^T\mathbf{\sigma}=\varepsilon^2$, $r=1$. 

\begin{figure}
    \centering
    \includegraphics[width=0.50\textwidth]{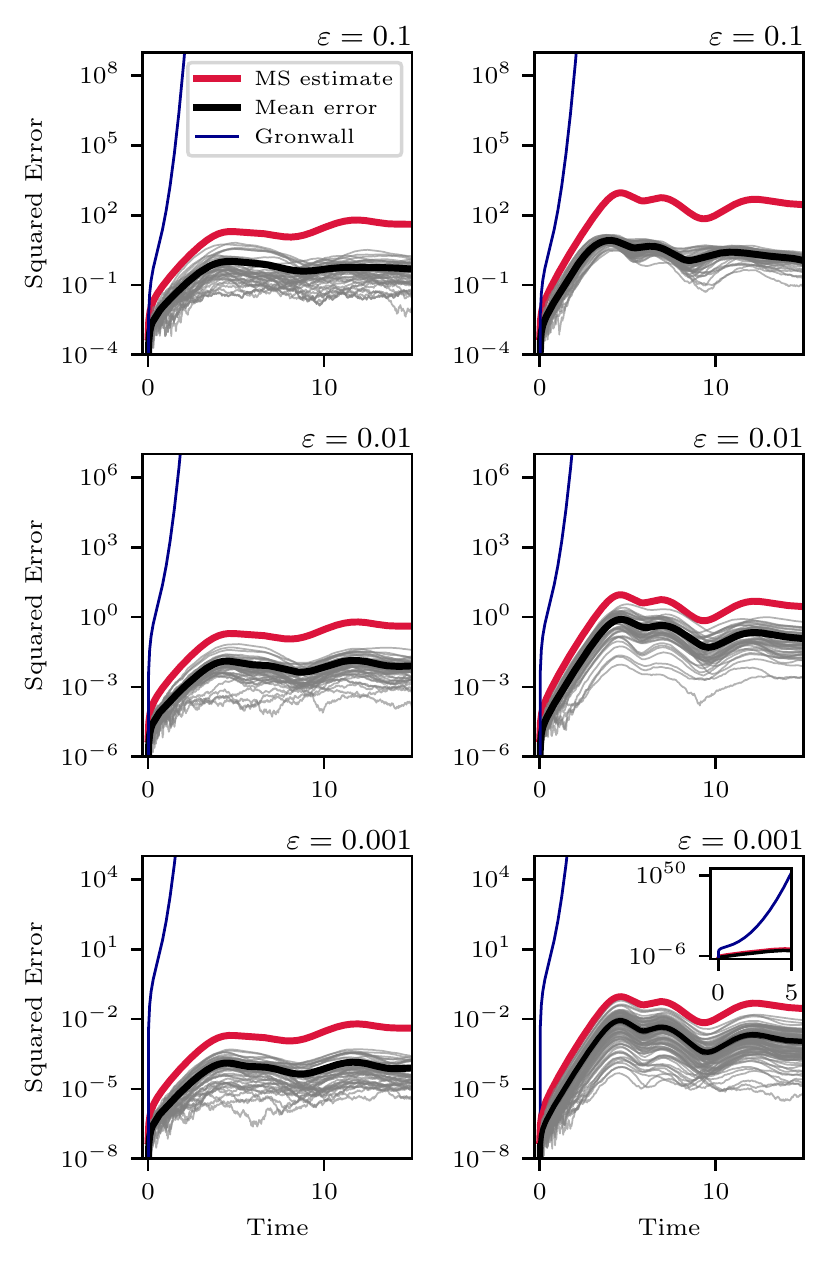}
    \caption{Square of the difference between the idealized and the real solutions to the Charney-DeVore model. Grey lines show the squared error along sample paths, the black curve is the average computed from 1000 sample paths. The red curve is the upper bound on the mean-squared leading-order trajectory uncertainty. The blue curve is the Gronwall-type upper bound for the mean-squared error\cite{note}, asymptotically given as $e^{L^2t^2/2}$, where $L$ is a Lipschitz-constant for \eqref{eq:cdv} and was chosen to be $L=1.8$. The inset shows the three curves on a larger scale. The trajectories start from the point $\mathbf{x}_{0}^{(1)} = (0, -0.012048, 0,-2.4217, 0,0)$ [$\mathbf{x}_{0}^{(2)} = (2.1084, 0, 0, -1.5904, 0,0)$] in the left [right] column. The magnitude of the perturbations is $\varepsilon$, which is indicated above the panels. $r=1$. }
    \label{fig:cdvmc}
\end{figure}

A comparison of the bound on the mean-squared leading-order trajectory uncertainty and the actual measured mean-squared trajectory uncertainty is shown in Fig. \ref{fig:cdvmc}, along with an appropriate Gronwall-type bound. Here, one of the initial conditions is chosen to lie on a steep ridge of MS (left column), while the other is chosen from a region with lower values (right column). The results show that the bound on the mean-squared leading-order trajectory uncertainty is respected for both initial conditions, in a wide range of $\varepsilon$. {While the mean-squared trajectory uncertainty is overestimated for the interval [0,15] (by a factor of around 10), the {\em trends} of the graph are captured accurately by our estimate in all of the examples shown}.

\section{\label{sec:ridges}Geometric structure of the model sensitivity}

Geometric descriptions of uncertainty in dynamical systems involve the finite-time Lyapunov exponent. This quantity describes
the growth rate of infinitesimal perturbations to
initial conditions. \emph{Ridges }of the FTLE field often signal repelling material surfaces in the phase space\cite{Haller2015}. Under further assumptions\cite{Haller2011b}, one can rigorously conclude the presence of a repelling hyperbolic LCS from an FTLE ridge. 

Our results show, that the FTLE
field in itself is not sufficient to characterize sensitivity to modeling errors in dynamical systems.
By Theorem 4, one needs to integrate the time-dependent FTLE
field over the time interval of interest, to obtain MS. Regardless, the two fields
are clearly related. 

For purely deterministic perturbations, let us
denote the maximal eigenvalue of the Cauchy-Green strain tensor, computed
over $[s,t]$ at the point $x_0$ by 

\begin{equation}
    \Lambda_{s,t}(x_0,t_0)=\Lambda_{s}^{t}\left(F_{t_{0}}^{s}(x_{0})\right),
\end{equation}

where $F_{t_{0}}^{s}(x_{0})$ is the flow-map of the idealized model \eqref{eq:model}. Specifically, with this notation, we can write the FTLE as

\begin{equation}
    \text{FTLE}_{t_{0}}^{t}(x_{0})=\frac{\log\sqrt{\Lambda_{t_0,t}(x_0,t_0)}}{t-t_{0}}.
\end{equation}

We will use quantity $\Lambda_{s,t}$ to connect features of the MS field to those of the FTLE field. Such a connection is already suggested by Fig. \ref{fig:ftlems}, which compares the FTLE field of the Charney-DeVore model to the field 
\begin{equation}
    \frac{1}{2(t-t_0)}\log \text{MS}_{t_0}^t(x_0,r)
\end{equation}
of the same model. That is, we display MS on a similar scale as the FTLE for a better comparison. This scale will be justified later. 

\begin{figure*}
\includegraphics[width=\textwidth]{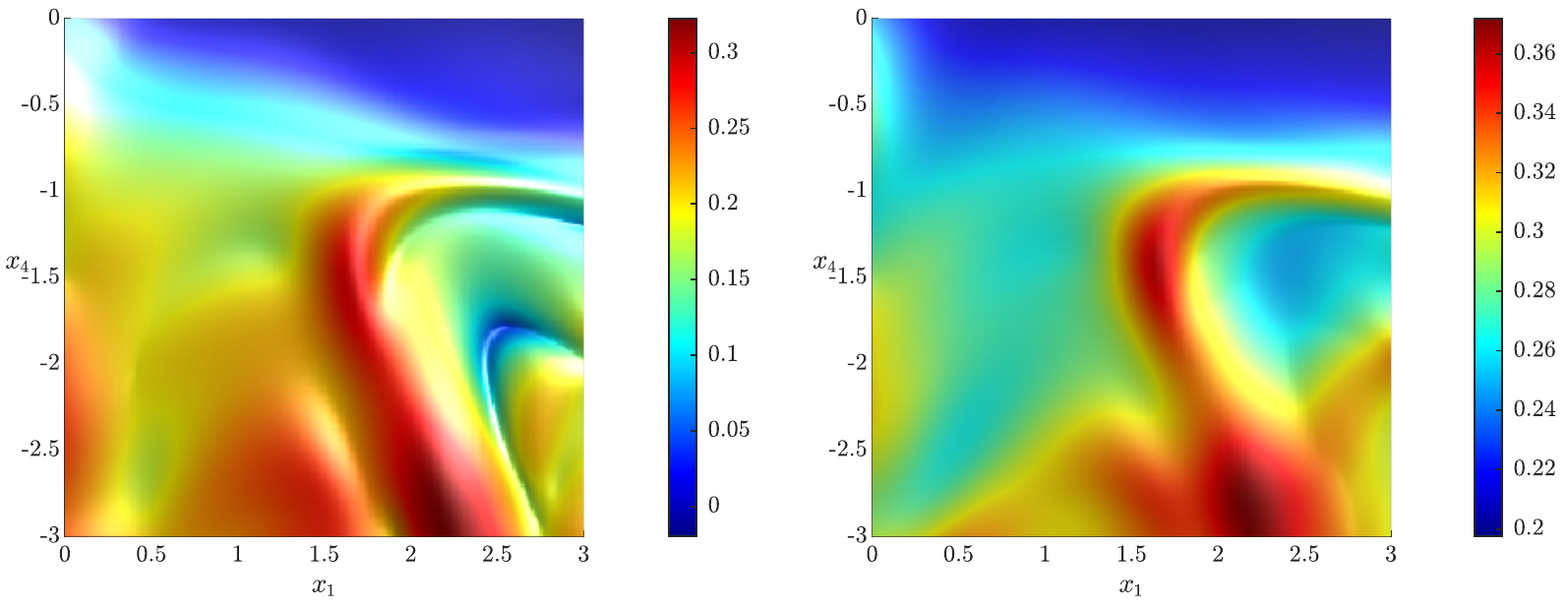} 
\caption{Comparison between the FTLE and the MS fields for the Charney-DeVore model. The scalar fields are shown for the time interval $[0,15]$, over the $x_1-x_4$ plane. In the left panel, the FTLE field, while in the right panel, the field $\frac{1}{2(t-t_0)}\log \text{MS}_{t_0}^t(x_0,0)$ is plotted.}
\label{fig:ftlems}
\end{figure*}

The figure shows that the main features of the FTLE field are also
found in the MS field, if they are compared over the
same time interval. Specifically, the main organizers of the dynamics, the FTLE-ridges, tend to persist in the MS field.
A closer look reveals, however, that this is not always the case. For example, in the region around $(x_1 = 2.5, x_4 = -1.5)$, finer ridges of the FTLE field disappear in the MS field.

To analyze this phenomenon, we adopt the following definition of a ridge from Ref. \onlinecite{Karrasch2013}. 
\begin{definition}
Let $f: \mathbb{R}^n \to \mathbb{R}$ be a smooth function and $M\subset \mathbb{R}^n$ a compact, codimension-one manifold with a boundary $\partial M$. The manifold $M$ is a \emph{ridge} of the scalar field $f$ if both $M$ and $\partial M$ are normally attracting invariant manifolds for the gradient system
\begin{equation}\label{eg:defridge}
    \dot{x} = \nabla f(x).
\end{equation}
\end{definition}
The term normally attracting invariant manifold\cite{Fenichel1971} refers to an invariant manifold for which contraction along the manifold is dominated by contraction normal to it. This allows the use of results that guarantee the persistence of ridges under small perturbations to the scalar field $f$.

To find a condition relating the MS- and FTLE-ridges, we assume that the deterministic modeling error in \eqref{eq:MSdef} is the only contributor to MS, that is, MS$_{t_0}^t(x_0,r)$ with $r =0$.
In this case, we can write
\begin{align}
    \text{MS}_{t_0}^t(x_0,0) &= \left(\int_{t_0}^t \sqrt{\Lambda_{s,t}(x_0,t_0)}ds \right)^2 \nonumber\\
    &= \Lambda_{t_0,t}(x_0,t_0) \left(\int_{t_0}^t \sqrt{\frac{\Lambda_{s,t}(x_0,t_0)}{\Lambda_{t_0,t}(x_0,t_0)}}ds \right)^2.
\end{align}
Taking the logarithm and using expression \eqref{eq:ftledef} for the FTLE field, we obtain
\begin{align}
\label{eq:msFTLE}
    \frac{\log \text{MS}_{t_0}^t(x_0,0)}{2(t-t_0)} = \text{FTLE}_{t_0}^t(x_0) + \frac{1}{t-t_0}\log\int_{t_0}^t \sqrt{\frac{\Lambda_{s,t}(x_0,t_0)}{\Lambda_{t_0,t}(x_0,t_0)}}ds.
\end{align}

Equation \eqref{eq:msFTLE} shows that we are able to write the appropriately scaled MS field as a perturbation of the FTLE field. The difference between the MS and the FTLE fields is shown in Fig. \ref{fig:fig7}, which suggests the \emph{values} of the two fields differ substantially, even at the location of the persisting ridge. We conclude that in general, ridges of the MS field are different from those of the FTLE field. 
As seen from \eqref{eq:msFTLE}, the MS field must be treated as a finite-size-perturbation to the FTLE field: the persistence results of Ref.  \onlinecite{Karrasch2013} do not apply. 

The general results on persistence of normally hyperbolic invariant manifolds \cite{Karrasch2013,Fenichel1971} state that for a ridge of a scalar field $f_0(x_0)$ to persist in the field $f(x_0)$, the appropriate gradient vector fields in \eqref{eg:defridge} must be $C^1$-$\theta$ close, for $\theta$ small enough. 

In our setting, this translates into a condition on the gradient of the difference field, defined by \eqref{eq:msFTLE}. This indicates that ridges of the FTLE field, along which 
\begin{align}
\label{eq:smallcond}
    \sup_{x_0\in V}&\left \Vert \nabla \frac{1}{t-t_0}\log \int_{t_0}^t \sqrt{\frac{\Lambda_{s,t}(x_{0},t_0)}{\Lambda_{t_0,t}(x_0,t_0)}}ds \right \Vert \leq \theta, \nonumber\\
    \sup_{x_0\in V}&\left \Vert \nabla^2 \frac{1}{t-t_0}\log \int_{t_0}^t \sqrt{\frac{\Lambda_{s,t}(x_{0},t_0)}{\Lambda_{t_0,t}(x_0,t_0)}}ds \right \Vert\leq \theta,
\end{align}
holds for $\theta$ sufficiently small, are expected to be close to ridges of the scalar field $\frac{\log \text{MS}_{t_0}^t(x_0,0)}{2(t-t_0)}$. 

\begin{figure}
    \centering
    \includegraphics[width=0.48\textwidth]{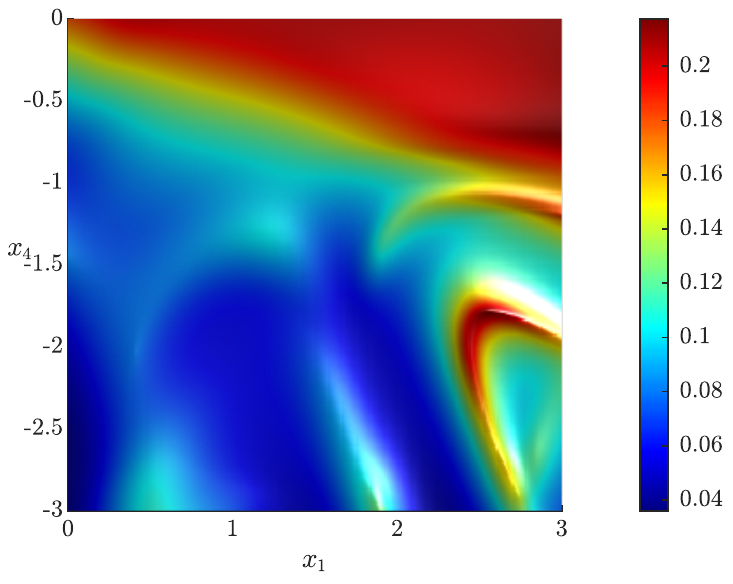}
    \caption{Difference between the MS and the FTLE fields in the Charney-DeVore model. The quantity $\frac{\log \text{MS}_{t_0}^t(x_0,0)}{2(t-t_0)} - \text{FTLE}_{t_0}^t(x_0)$ is shown over the $x_1-x_4$ plane, for the time interval $[0,15]$. }
    \label{fig:fig7}
\end{figure}

\section{Conclusions}
We have investigated the effect of modeling uncertainties on trajectories of a dynamical system. Under general smoothness assumptions, in a deterministic setting, we derived a bound on the leading order trajectory uncertainty which can be computed using the idealized model dynamics and assuming a bound on the magnitude of the modeling error. Our upper bound depends on the eigenvalues of an appropriate Cauchy-Green strain tensor of the idealized model, allowing for a location-specific assessment of trajectory uncertainty in the phase space.

We have also generalized our result to the case of stochastic modeling errors. In that setting, we have introduced the \emph{Model Sensitivity (MS)}, a coefficient relating the modeling uncertainty to the bound of the mean-squared leading-order trajectory uncertainty. This MS is computed solely in terms of the idealized, deterministic dynamics as a general functional of the invariants of the Cauchy-Green strain tensor. As a consequence, we do not need to assume any specific form for the modeling errors to quickly assess their effects. Contrary to prior, statistical and data based methods, MS quantifies trajectory sensitivity based on the dynamical properties of the known model.

We have also shown that our bounds on the leading order trajectory uncertainty are optimal. Specifically, for a class of linear systems, we gave an example in which the mean-squared trajectory uncertainty was exactly equal to the product of the MS and the modeling uncertainty. Therefore, for general systems satisfying our smoothness assumptions, our bounds cannot be improved. 

{On numerical examples (one of which was a chaotic system), we showed that MS can be a useful predictor of local trajectory uncertainty: the mean-squared trajectory uncertainty qualitatively follows the bound defined by \eqref{eq:leadbound} for surprisingly long time intervals, which is not the case with the classical, Gronwall-type bounds\cite{Kirchgraber1976,note}. When viewed as a scalar field over the phase space, the MS field can exhibit complex structure, which allows us to distinguish particularly sensitive regions.} We have shown that the MS fields, are similar to the FTLE fields, which are often used to characterize instability in phase space. This is in line with the usual reasoning that the instabilities within a dynamical system typically grow with the rate of the largest Lyapunov exponent. Our results make this argument precise by pointing out the exact relationship between MS and the FTLE. In particular, we find that not all features of the FTLE field persist in the MS field.

The sensitivity analysis developed here becomes more computationally intensive for higher-dimensional dynamical systems. Indeed, the calculation of the Cauchy-Green strain tensor becomes problematic for even a few hundred dimensions. This could be improved by using approximate methods such as OTD (Optimally Time Dependent) modes\cite{Babaee2017}, enabling the calculation of the dominant Cauchy-Green eigenvalue with much less effort. This approach could give useful results even for certain climate models, for which sensitivity analysis is critical.
\section*{Acknowledgements}
We are grateful to Hessam Babaee for his help with the implementation of the Charney-DeVore model. 
\section*{Data Availability Statement}
The data that support the findings of this study are openly available in the repository Ref. \onlinecite{github}.
It contains the numerical implementation of the examples discussed in the text.
\section{Appendix}
\subsection{\label{sec:app1}Proof of Theorem 1}
The trajectory uncertainty, based on a power-series expansion, is
\begin{equation}
   \left \vert x^0(t)-x^\varepsilon(t)\right \vert = |\varepsilon \eta(t,t_0,x_0) + O(\varepsilon^2)|. 
\end{equation}
This can be bounded by
\begin{equation}
\label{eq:proof1}
    |\varepsilon \eta(t,t_0,x_0) + O(\varepsilon^2)| \leq \varepsilon|\eta(t,t_0,x_0)| + |O(\varepsilon^2)|, 
\end{equation}
where the remainder $O(\varepsilon^2)$ term can be bounded by $M\varepsilon^2$, when $\varepsilon \leq \bar{\varepsilon}$ for some $\bar{\varepsilon}>0$. This bound depends on $x_0$, but assuming a compact domain $U$ within $\mathbb{R}^n$, we can choose the constants $M>0$ and $\bar{\varepsilon}$ such that this bound is satisfied for all $x_0\in U$. 

Next, we simply substitute the bound \eqref{eq:sensitivitybound} for the leading-order trajectory uncertainty into Eq. \eqref{eq:proof1} to obtain
\begin{align}
\label{eq:proof2}
&\left \vert x^0(t)-x^\varepsilon(t) \right\vert \nonumber\\
&\;\leq \left(\int_{t_0}^t \sqrt{\Lambda_s^t\left(x^0(s)\right)}ds \right)\Delta_\infty(x_0,t) + M\varepsilon^2\nonumber\\
&\;=\left(\int_{t_0}^t \sqrt{\Lambda_s^t\left(x^0(s)\right)}ds \right)\Delta_\infty(x_0,t) + \Delta_\infty(x_0,t) \frac{M\varepsilon^2}{\Delta_\infty(x_0,t)} \nonumber\\
&\;= \left(\int_{t_0}^t \sqrt{\Lambda_s^t\left(x^0(s)\right)}ds +\frac{M\varepsilon^2}{\Delta_\infty(x_0,t)} \right)\Delta_\infty(x_0,t)
\end{align}
Comparing Eq. \eqref{eq:proof2} to Eq. \eqref{eq:uncertbound} we obtain that if \begin{equation}
    \delta<\min_{\substack{t\in{[t_0,t_1]}\\x_0\in U}}\frac{M\bar{\varepsilon}^2}{\Delta_\infty(x_0,t)}=\bar{\delta},
\end{equation}
then we can set
\begin{equation}
    \varepsilon_0 := \max_{\substack{t\in{[t_0,t_1]}\\x_0\in U}}\sqrt{\frac{\Delta_\infty(x_0,t)\delta}{M}}.
\end{equation}
Otherwise, if $\delta\geq \bar{\delta}$, inequality \eqref{eq:proof2} is satisfied for $\varepsilon_0:=\bar{\varepsilon}$. Hence, for all $\delta>0$ we can choose
\begin{equation}
    \varepsilon_0:= \min\left\{\max_{\substack{t\in{[t_0,t_1]}\\x_0\in U}}\sqrt{\frac{\Delta_\infty(x_0,t)\delta}{M}},\bar{\varepsilon}\right\}
\end{equation}
as claimed.
\subsection{\label{sec:app2}Proof of Theorem 2}

The statement follows from an asymptotic
expansions for stochastic differential equations \citep{blago1962,Friedlin2012}
of the form 
\begin{equation}
dX_{t}^{\varepsilon}=\mu(X_{t}^{\varepsilon})dt+\Sigma^{\varepsilon}(X_{t}^{\varepsilon})dW_{t}.\label{eq:autonomousSDE}
\end{equation}
Assume that the coefficient functions $\mu$ and $\Sigma^{\varepsilon}$
have bounded and measurable derivatives up to order $m$. Then, there exists $\bar{\varepsilon}>0$, such that for $\varepsilon<\bar{\varepsilon}$ one can
recursively obtain stochastic differential equations for the stochastic
variables $X_{t}^{(k)}$. For any $k<m$
\begin{equation}
    X_{t}^{\varepsilon}=X_{t}^{0}+\varepsilon X_{t}^{1}+\varepsilon^{2}X_{t}^{2}+...+\varepsilon^{k-1}X_{t}^{k-1}+ \varepsilon^k R_{k}(t,\varepsilon),
\end{equation}
with the remainder term being bounded in the mean-squared sense. For a proof of the $n$ dimensional case, see Ref.~\onlinecite{Albeverio2015}.

To generalize this result for time- and $\varepsilon$-dependent drift and diffusion
coefficients, we proceed by introducing an SDE on the extended phase
space $\mathbb{R}^{n}\times[t_{0},t_{1}]\times[0,\bar{\varepsilon}]$ of Eq. (\ref{eq:SDE}).
Let $X^\varepsilon_{t}\in$$\mathbb{R}^{n}\times[t_{0},t_{1}]\times[0,\bar{\varepsilon}]$, $X^\varepsilon_{t}=(x_{t,1},x_{t,2},...,x_{t,n},t, \varepsilon)$ and

\begin{align}
    \mu(X^\varepsilon_{t})&=\begin{pmatrix}f_0(x_{t},t)_{1}+\varepsilon g(x_{t},t,\varepsilon)_1\\
f_0(x_{t},t)_{2}+\varepsilon g(x_{t},t,\varepsilon)_2\\
\vdots\\
f_0(x_{t},t)_{n}+\varepsilon g(x_{t},t,\varepsilon)_n\\
1 \\
0
\end{pmatrix}, \nonumber\\
\Sigma^{\varepsilon}(X^\varepsilon_t)&=\begin{pmatrix}\varepsilon\sigma_{11}(x_{t},t) & \cdots & \varepsilon\sigma_{1n}(x_{t},t) & 0 &0\\
\vdots & \ddots & \vdots & 0 & 0\\
\varepsilon\sigma_{n1}(x_{t},t) & \cdots & \varepsilon\sigma_{nn}(x_{t},t) & 0 & 0 \\
0 & \cdots & 0 & 0 & 0 \\
0 & \cdots & 0 & 0 & 0
\end{pmatrix}.
\end{align}

The resulting SDE has the desired form of Eq. (\ref{eq:autonomousSDE})
and the coefficients retain the analytic properties of the functions
$f$ and $\sigma$, i. e. they remain measurable and have bounded
derivatives. This means we can apply the result of Ref. \onlinecite{Albeverio2015}.
to obtain the following first-order expansion

\begin{equation}
    X_{t}^{\varepsilon}=X_{t}^{0}+\varepsilon X_{t}^{1}+\varepsilon^2 R_{2}(t,\varepsilon)
\end{equation}
for the solutions. The coefficients in the expansion are governed by the following set
of linear SDEs:

\begin{align}
dX_{t}^{0}=&\mu(X_{t}^{0})\,dt, \nonumber\\
X_{t=t_{0}}^{0}=&(x_{0},t_{0},0),\nonumber\\
dX_{t}^{1}=&\nabla \mu(X_{t}^{0})X_{t}^{1}\,dt+\left.\frac{\partial \Sigma^{\varepsilon} }{\partial \varepsilon} \right\vert_{X_t^0}\,dW_{t},\nonumber \\
 X_{t=t_{0}}^{1}=&(0,t_0, 0)\nonumber.
\end{align}

Setting $X_{t}^{0}=(x^0_t,t,0),\:X_{t}^{1}=(\eta_{t},t,0)$
and keeping only the the first $n$ entries of the vectors
yields the following expansion for the nonautonomous system \eqref{eq:SDE}:
\begin{align}
\label{eq:SDEfirstorder}
dx^0_t =&f_{0}\left(x^0_t,0\right)\,dt, \nonumber\\
\quad x^0_{t=t_0}=&x_{0},\nonumber\\
d\eta_{t} =&\nabla f_{0}\left(x^0_t,t\right)\eta_{t}\,dt +g\left(x^0_t,t;0\right)\,dt\nonumber\\
 &+\sigma\left(x^0_t,t\right)\,dW_{t},\nonumber \\
  \eta_{t=t_{0}}=&0,
\end{align}
as claimed.
\subsection{\label{sec:app3}Proof of Theorem 3}

We seek a solution of the inhomogeneous, linear SDE (\ref{eq:smallnoise2})
using the method of 'variation of constants' on the solution of the
homogeneous equation. Let the solution of the corresponding homogeneous
equation be
\begin{equation}
    x_{H}(t,x_{0})=\varphi(t)x_{0}.
\end{equation}

Here, $\varphi(t)$ is the (linear) flow map of Eq. \eqref{eq:eq of variations}, mapping initial conditions
at time $t_{0}$ to their position at time $t$. By the method of
variation of constants, let $\eta_{t}$ be of the form $\eta_{t}=\varphi(t)x_{t}$
for some random variable $x_{t.}$ We now compute the differential
$d\eta_{t}$, keeping in mind that $\eta_{t}$ is a vector-valued
stochastic process, requiring the use of It\^o's formula. However,
since $\eta_{t}=\eta(t,x)$ is only linear in the $x$ variable, we
simply have 
\begin{equation}
\label{eq:proof3}
    d\eta_{t}=\dot{\varphi}(t)x_{t}dt+\varphi(t)dx_{t}.
\end{equation}
Substituting Eq. (\ref{eq:smallnoise2}) into Eq. \eqref{eq:proof3} and noting that $\varphi$
is the fundamental matrix-solution to the equation of variations \eqref{eq:eq of variations} yields
\begin{align}
\label{eq:proof4}
d\eta_{t}  =&\dot{\varphi}(t)x_{t}\,dt+\varphi(t)dx_{t}=\nabla f_{0}\left(x^0_t,t\right)x_{t}\,dt+\varphi(t)dx_{t},\nonumber\\
\varphi(t)dx_{t}  =&g\left(x^0_t,t;0\right)\,dt+\sigma\left(x^0_t,t\right)\,dW_{t},\\
dx_{t}  =&\varphi(t)^{-1}\left[g\left(x^0_t,t;0\right)\,dt+\sigma\left(x^0_t,t\right)\,dW_{t}\right].\nonumber
\end{align}

The last expression in Eq. \eqref{eq:proof4} is an It\^o-integral, which can be evaluated as
\begin{equation}
   x_{t}=\int_{t_{0}}^{t}\varphi(s)^{-1}g\left(x^0_s,s;0\right)\,ds+\int_{t_{0}}^{t}\varphi(s)^{-1}\sigma\left(x^0_{s},s\right)\,dW_{s}.
\end{equation}
Using the form of $\eta_{t}$ and observing that $\phi_{s}^{t}\left(x^0_{s}\right)=\varphi(t)\varphi(s)^{-1}$
is the \emph{normalized} fundamental matrix solution to Eq. (\ref{eq:eq of variations}),
we obtain
\begin{align}
\eta_{t} =&\int_{t_{0}}^{t}\varphi(t)\varphi(s)^{-1}g\left(x^0_s,s;0\right)\,ds\nonumber\\
&+\int_{t_{0}}^{t}\varphi(t)\varphi(s)^{-1}\sigma\left(x^0_{s},s\right)\,dW_{s}\\
 =&\int_{t_{0}}^{t}\phi_{s}^{t}\left(x^0_{s}\right)g\left(x^0_s,s;0\right)\,ds+\int_{t_{0}}^{t}\phi_{s}^{t}\left(x^0_{s}\right)\sigma\left(x^0_{s},s\right)\,dW_{s},\nonumber
\end{align}
which proves the statement of Eq. \eqref{eq:eta_stoch}.
\subsection{\label{sec:app4}Proof of Theorem 4}

First, we compute $\EX(N^{2})$. By the properties of the It\^o-integral,
the expected value of the mixed term in Eq. \eqref{eq:nsquared} is 0, and hence
\begin{align}
\label{eq:proof5}
\EX(&N^{2}) \nonumber\\
 =& \EX(N_{d}^{2})+\EX(N_{s}^{2})+2\EX(N_{m})\nonumber\\
 =&N_{d}^{2}+\EX(N_{s}^{2})\nonumber\\
 & +2\left(\int_{t_{0}}^{t}\phi_{s}^{t}\left(x^0_{s}\right)g\left(x^0_s,s;0\right)\,ds\right)\EX\left(\int_{t_{0}}^{t}\phi_{s}^{t}\left(x^0_{s}\right)\sigma\left(x^0_{s}s\right)\,dW_{s}\right)\nonumber\\
 =&N_{d}^{2}+\EX(N_{s}^{2}).
\end{align}

For the stochastic part of the mean-squared leading-order trajectory uncertainty, we utilize
It\^o's isometry component-wise to obtain
\begin{align}
&\EX(N_{s}^{2})\nonumber\\
&\;=\EX\left[\left(\int_{t_{0}}^{t}\phi_{s}^{t}\left(x^0_{s}\right)\sigma\left(x^0_{s},s\right)\,dW_{s}\right)^{2}\right] \nonumber\\
&\;=\EX\left[\left(\int_{t_{0}}^{t}\phi_{s}^{t}\left(x^0_{s}\right)\sigma\left(x^0_{s},s\right)\,dW_{s}\right)\left(\int_{t_{0}}^{t}\phi_{s}^{t}\left(x^0_{s}\right)\sigma\left(x^0_{s},s\right)\,dW_{s}\right)\right] \nonumber\\
&\;=\EX\left[\sum_{i,j,k,l,m}\left(\int_{t_{0}}^{t}(\phi_{s}^{t})_{ij}\sigma_{jk}\,(dW_{s})_{k}\right)\left(\int_{t_{0}}^{t}(\phi_{s}^{t})_{il}\sigma_{lm}\,(dW_{s})_{m}\right)\right]\nonumber\\
&\;=\sum_{i,j,k,l,m}\EX\left[\left(\int_{t_{0}}^{t}(\phi_{s}^{t})_{ij}\sigma_{jk}(\phi_{s}^{t})_{il}\sigma_{lm}\,\left[(dW_{s})_{k},(dW_{s})_{m}\right]\right)\right]
\end{align}
The notation $\left[(dW_{s})_{k},(dW_{s})_{m}\right]$ refers to the quadratic covariation\cite{Kallenberg97} of the processes $(dW_{s})_{k}$ and $(dW_{s})_{m}$. Since the components of the $n$-dimensional Wiener-process are assumed to be independent, we have (by It\^o's isometry),
\begin{align}
\EX& \left(\int_{t_{0}}^{t}(\phi_{s}^{t})_{ij}\sigma_{jk}(\phi_{s}^{t})_{il}\sigma_{lm}\,\left[(dW_{s})_{k},(dW_{s})_{m}\right]\right) \nonumber \\ &=\EX\left(\int_{t_{0}}^{t}(\phi_{s}^{t})_{ij}\sigma_{jk}(\phi_{s}^{t})_{il}\sigma_{lm}\delta_{km}\,ds\right),
\end{align}
where $\delta_{km}$ is the Kronecker-delta.
Denoting the Frobenius-norm by $||\cdot||_{F}:\mathbb{R}^{n\times n}\to\mathrm{\mathbb{R}^{+}}$, we have
$||A||_{F}^{2}=\sum_{i,j}|A_{ij}|^{2}=\text{tr}(A^{T}A)$. Therefore,
\begin{align}
\label{eq:proof6}
    \EX(N_{s}^{2})&=\int_{t_{0}}^{t}||\phi_{s}^{t}\sigma||_{F}^{2}\,ds \nonumber\\
    \varepsilon^{2}\EX\left[N_{s}(t)^{2}\right]&\leq\int_{t_{0}}^{t}\varepsilon^{2}||\phi_{s}^{t}||_{F}^{2}\,ds\max_{s\in[t_{0},t]}\left \Vert\sigma\left(x^0_{s},s\right)\right \Vert_{F}^{2}\nonumber\\
    &=\int_{t_{0}}^{t}\mathrm{tr}\left[C_{s}^{t}\left(x^0_{s}\right)\right]\,ds\,\Delta_{\infty}^{\sigma}(x_{0},t).
\end{align}

In Section III. A, we also concluded in Eq. \eqref{eq:sensitivitybound} that $\delta(x_{0},t)=\varepsilon N_{d}(t)\leq\int_{t_{0}}^{t}\sqrt{\Lambda_{s}^{t}(x_{s})}\,ds\,\Delta_{\infty}(x_{0},t).$
Substituting Eqs.~\eqref{eq:sensitivitybound} and \eqref{eq:proof6} into Eq. \eqref{eq:proof5} implies
\begin{align}
    \varepsilon^{2}\EX\left[N(t)^{2}\right]\leq&\left(\int_{t_{0}}^{t}\sqrt{\Lambda_{s}^{t}\left(x^0_{s}\right)}\,ds\right)^{2}\,\Delta_{\infty}^{2}(x_{0},t)\nonumber\\
    &+\int_{t_{0}}^{t}\mathrm{tr}\left[C_{s}^{t}\left(x^0_{s}\right)\right]\,ds\,\Delta_{\infty}^{\sigma}(x_{0},t),
\end{align}
as claimed.
\subsection{\label{sec:app5}Proof of Theorem 5}
Using the small-noise expansion \eqref{eq:smallnoise} for the mean-squared trajectory uncertainty, for $\varepsilon< \bar{\varepsilon}$, we obtain
\begin{align}
\label{eq:proof7}
    \EX&\left( |x_t^\varepsilon-x^{0}(t)|^2\right) \nonumber\\
    &=\EX \left( |\varepsilon \eta_t +  \varepsilon^2 R(t,\varepsilon)|^2\right). 
\end{align}
Using the Minkowski-inequality, we also find that 
\begin{align}
    &\sqrt{\EX\left( |x_t^\varepsilon-x^{0}(t)|^2\right)} \nonumber\\
    & \leq \sqrt{\varepsilon^2 \EX(|\eta_t|^2)} +  \sqrt{\varepsilon^4\EX(|R(t,\varepsilon)|^2)}.
\end{align}
Since the second order remainder term in Eq. \eqref{eq:smallnoise} is bounded in the mean-squared sense, we have, for some $K_0< \infty$,
\begin{equation}
\label{eq:proof8}
\sup_{t\in[t_0,t_1]}\EX(|R(t,\varepsilon)|^2)\leq K_0^2.
\end{equation}
To bound $\EX(|\eta_t|^2)$, we use Theorem 4 in the form of Eq. \ref{eq:MSdef} to obtain
\begin{equation}
\label{eq:proof9}
    \varepsilon^2 \EX(|\eta_t|^2) \leq \text{MS}_{t_0}^t(x_0,r)\Delta_\infty^2(x_0,t). \\
\end{equation}

Substituting bounds \eqref{eq:proof8} and \eqref{eq:proof9} into the original expression \eqref{eq:proof7}, we have

\begin{equation}
\label{eq:proof10}
    \sqrt{\EX\left( |x_t^\varepsilon-x^{0}(t)|^2\right) }\leq \sqrt{\text{MS}_{t_0}^t(x_0,r)\Delta_\infty^2(x_0,t)} + \sqrt{K_0^2\varepsilon^4} 
\end{equation} 

Since $x_0$ is taken from a compact domain $U\subset \mathbb{R}^n$, we can choose the constant $K_0$ to be independent of $x_0$. After rearranging the terms, we obtain
\begin{align}
    \label{eq:proof11}
    &\sqrt{\EX\left( |x_t^\varepsilon-x^{0}(t)|^2\right)}\nonumber\\
    &\leq\sqrt{\text{MS}_{t_0}^t(x_0,r)}\Delta_\infty(x_0,t) + \Delta_\infty(x_0,t) \frac{K_0\varepsilon^2 }{\Delta_\infty(x_0,t)} \nonumber\\
    &= \left(\sqrt{\text{MS}_{t_0}^t(x_0,r)} +  \frac{K_0\varepsilon^2 }{\Delta_\infty(x_0,t)}\right)\Delta_\infty(x_0,t).
\end{align}
Comparing \eqref{eq:theorem5} to \eqref{eq:proof11}, we obtain the statement of Theorem 5 after setting  
\begin{equation}
    \varepsilon_0:= \min\left\{\max_{\substack{t\in{[t_0,t_1]}\\x_0\in U}}\sqrt{\frac{\Delta_\infty(x_0,t)\delta}{K_0}},\bar{\varepsilon}\right\}.
\end{equation}


\begin{thebibliography}{61}
\bibitem [{not()}]{note}%
  \BibitemOpen
  \href@noop {} {}\bibinfo {howpublished} {The mean-square of the trajectory
  uncertainty, $\mathbb{E}\left(|x^\varepsilon_t -x^0_t|^2\right)$, can be
  upper estimated by standard tools \cite{Oksendal2010}, using the
  Lipschitz-property of $f_0(x,t)$ and the bounds for the functions
  $g(x,t,\varepsilon)$ and $\sigma(x,t)$. The estimate is of the form $v(t)\leq
  k(t)\int_{t_0}^t v(s) ds + h(t)$, which is a suitable form for applying
  Gronwall's lemma\cite{Arnold1992}. In our formulation, the bound can finally
  be expressed as $\mathbb{E}\left(|x^\varepsilon_t -x^0_t|^2\right)\leq
  3\Delta^2_\infty(t-t_0)^2 +
  3\Delta^\sigma_\infty(t-t_0)+e^{\frac{L^2}{2}(t-t_0)^2}-1$, where $L$ is a
  Lipschitz-constant for $f_0$, $\Delta_\infty$ and $\Delta^\sigma_\infty$ are
  the deterministic and stochastic modeling uncertainties,
  respectively.}\BibitemShut {Stop}%
\end{thebibliography}
\end{document}